\documentclass[preprint,times,12pt]{elsarticle}

\usepackage[T1]{fontenc}

\usepackage{subcaption} 

\usepackage{multirow} 
\usepackage{booktabs} 

\usepackage{nicematrix} 

\usepackage[tworuled,algo2e]{algorithm2e} 

\usepackage{tikz}
\usetikzlibrary{shapes,arrows} 

\usepackage{cleveref}

\usepackage{url} 
\usepackage{geometry}
\usepackage{bm} 
\usepackage{amsthm} 
\usepackage{dsfont} 
\DeclareMathAlphabet\mathbfcal{OMS}{cmsy}{b}{n}
\foreach \x in {a, ..., z}{
    \expandafter\xdef\csname bf\x \endcsname{\noexpand\ensuremath{\noexpand\mathbf{\x}}} 
}
\foreach \x in {A, ..., Z}{
    \expandafter\xdef\csname bf\x \endcsname{\noexpand\ensuremath{\noexpand\mathbf{\x}}} 
    \expandafter\xdef\csname ds\x \endcsname{\noexpand\ensuremath{\noexpand\mathds{\x}}} 
    \expandafter\xdef\csname cal\x \endcsname{\noexpand\ensuremath{\noexpand\mathcal{\x}}} 
    \expandafter\xdef\csname bcal\x \endcsname{\noexpand\ensuremath{\noexpand\mathbfcal{\x}}} 
}
\newcommand{\vUn}{\mathds{1}} 

\newcommand{\out}{\mathop{\otimes}} 
\newcommand{\khatri}{\odot} 
\newcommand{\hadam}{\boldsymbol{\ast}} 
\newcommand{\T}{{\sf T}} 
\DeclareMathOperator*{\trace}{tr} 
\DeclareMathOperator*{\argmin}{\mathrm{argmin}} 
\def\p{{\mathrm{p}}} 
\def\dist{{\mathrm{D}}} 
\def\LW{{\mathrm{LW}}} 
\def\Diag{{\mathrm{Diag}}} 
\DeclareMathOperator*{\vectorize}{vec} 
\DeclareMathOperator*{\Card}{Card} 
\def\Error{{\mathrm{Err}}} 
\newcommand{\Err}[1]{\Error_{#1}} 
\newcommand{\triples}[3]{\left\{#1,#2,#3\right\}} 
\newcommand{\jkl}{\triples{j}{k}{\ell}} 
\newcommand{\Z}[1]{Z^{(#1)}} 
\newcommand{\im}[1]{i_{#1}} 
\newcommand{\Am}[1]{\bfA^{(#1)}} 
\newcommand{\Amh}[1]{\widehat{\bfA}^{(#1)}} 
\newcommand{\am}[2]{\bfa^{(#1)}_{#2}} 
\newcommand{\amh}[2]{\widehat{\bfa}^{(#1)}_{#2}} 
\newcommand{\amr}[1]{\am{#1}{r}} 
\newcommand{\amrh}[1]{\amh{#1}{r}} 
\newcommand{\lbd}{\bm{\lambda}} 
\newcommand{\lbdh}{\widehat{\bm{\lambda}}} 
\newcommand{\cpdM}{\cpdsetp{\lbd; \Am{1},\ldots, \Am{M}}} 
\newcommand{\cpdT}{\cpdsetp{\lbd; \Am{j},\Am{k}, \Am{\ell}}} 
\newcommand{\cpdsetp}[1]{\left[\!\left[#1\right]\!\right]} 
\newcommand{\Htrip}[3]{\bcalH^{(#1#2#3)}}
\newcommand{\Htript}[3]{\widetilde{\bcalH}^{(#1#2#3)}}
\newcommand{\Hjkl}{\Htrip{j}{k}\ell}
\newcommand{\Hjklt}{\widetilde{\bcalH}^{(jk\ell)}}

\newtheorem{definition}{Definition}[section]
\newtheorem{lemma}[definition]{Lemma}
\newtheorem{remark}[definition]{Remark}
\newtheorem{example}[definition]{Example}

\journal{Signal Processing}

\begin{document}

\begin{frontmatter}

\title{Coupled tensor models for probability mass function estimation: Part I, Principles and algorithms.}

\author{Philippe FLORES$^{a,1}$, Konstantin USEVICH$^{2}$, David BRIE$^{2}$}

\affiliation{organization={Corresponding author -- $^{1}$CNRS, Universit\'e Grenoble Alpes, Grenoble INP, GIPSA-lab ; 11 rue des Mathématiques, 38402 Saint-Martin-d'Hères, France -- $^{2}$CNRS, Universit\'e de Lorraine, CRAN ; Campus Sciences BP 70239, Vandœuvre-lès-Nancy, France}} 

\begin{abstract} 
    In this article, a Probability Mass Function (PMF) estimation method which tames the curse of dimensionality is proposed.
    This method, called Partial Coupled Tensor Factorization of 3D marginals or PCTF3D, has for principle to partially couple order-3 data projections -- seen as order-3 tensors -- to obtain a tensor decomposition of the probability mass tensor.
    The novelty of PCTF3D relies on partial coupling which consists in choosing a subset of 3D marginals.
    The choice of marginals is then formulated with hypergraphs.
    After presenting possible coupling strategies, some numerical experiments and an application of the method are proposed.
    This article is the first of a two-part article.
    While this first article focuses on a new algorithmic framework for PMF estimation, the second studies uniqueness properties of the model introduced in this article.
\end{abstract}


\end{frontmatter}

\section{Introduction} \label{sec:intro}

The problem of probability mass function (PMF) estimation arises in numerous signal processing, machine learning and data analytics.
A first example is the supervised learning problem that aims at grouping labeled data samples into classes.
Knowing the joint PMF of the random variables and associated labels is often seen as the "gold standard" that permits to affect samples to a class thanks to the Maximum a posteriori principle \cite{duda_pattern_1973,van_trees_detection_2004}.
A second example is in recommender systems which aims at recommending items to users according to the predicted rating.
Again knowing the PMF makes this task readily implementable.

A naive approach to estimate the PMF consists in counting the number of occurrences of the random variable realizations.
This method is useful for low-dimensional problems (roughly when the number of variables is $\leq 3$) but it hardly scales to higher dimensions because of the curse of dimensionality, which states that the complexity of a problem increases exponentially with the number of dimensions.
In fact, when no additional information on the variables, such as tree or Markovian structures, is available, estimating the PMF is often considered as unfeasible.
A very promising approach was recently proposed by \cite{n_kargas_tensors_2018} in which the high dimensional PMF is modeled as a naive Bayes model (NBM) resulting in a constrained CPD model whose parameters are estimated by coupling order-3 (or 4) marginals - seen as order 3 or 4 tensors.
The point is that these low-order marginals can be estimated with a limited number of samples.
This approach falls into the general framework of coupled tensor decompositions which has proved to be effective in a number of applications such as hyperspectral super-resolution \cite{kanatsoulis_hyperspectral_2018, prevost_hyperspectral_2022}, multitask fMRI analysis \cite{borsoi:hal-04135339}, data mining and data fusion \cite{ermics2015link, papalexakis2016tensors}.
In short, coupled tensor decompositions allow modeling multidimensional data with low complexity and enjoy provably theoretical uniqueness guarantees under mild conditions.
This holds true for PMF estimation by coupled tensor factorizations of order-3 marginals.
In \cite{ibrahim_recovering_2021}, the coupling of 2D marginals -- seen as matrices -- to obtain an estimation of the higher-order PMF is addressed.
While this approach is very appealing since the number of samples required to estimate 2D marginals can be reduced (as compared to 3D marginals), it appears that, similarly to non-negative matrix factorization, uniqueness guarantees require additional constraints to be fulfilled.
In \cite{a_yeredor_maximum_2019}, an alternative approach is proposed.
It is also based on the naive Bayes model, but it relies on a Kullback-Leibler Divergence optimization approach and does not require estimating order-3 marginals.
The factors of the naive Bayes model are estimated directly from the data leading to a computational burden which hardly scales to large amounts of data.
This approach was extended by \cite{kargas2019learning} to smooth probability density function (PDF) estimation using a two-steps approach : jointly factorize histogram estimates (\emph{i.e.} PMFs) of lower-dimensional PMFs and interpolate the factors to recover the PDF.
A very interesting point is that this approach has proved to outperform by far the classical EM for Gaussian mixture models \cite{reynolds_gaussian_2009,hsu_learning_2013}.
Our conviction is that this approach may also be relevant in other data analytics and signal processing applications where PDF estimation is central such as sensor networks \cite{nowak_distributed_2003}, biology \cite{gyllenberg_non-uniqueness_1994} and social network studies \cite{nowicki_estimation_2001}. 

In this paper we will restrict ourselves to the PMF estimation problem, keeping in mind that it can be straightforwardly extended to the PDF estimation problem following \cite{kargas2019learning}.
A key aspect of the tensor-based PMF estimation relies on the ability of the naive Bayes model to break the curse of dimensionality since its complexity remains linear with the number of dimensions.
However, the computational burden of these methods may limit their practical use.
This is the case for \cite{a_yeredor_maximum_2019} whose computational complexity is directly affected by the number of dimensions and the number of samples.
Working with 3D-marginals as \cite{n_kargas_tensors_2018} may yield a lower computational burden, but it remains strongly impacted by the number of dimensions, noted $M$, since the number of 3D marginals is $\binom{M}3$, hence increases cubically with $M$.
For example, if one wants to estimate an order-($M=20$) PMF with $I=15$ bins per dimension, one has to solve an optimization problem coupling $1140$ marginals of size $15^3 = 3375$.

In this paper, the Partially Coupled Tensor Factorization of 3D marginals or PCTF3D is proposed.
The idea is very simple : instead of coupling all the 3D-marginals, only a limited number of marginals are used to estimate the NBM.
By choosing the number of marginals to be coupled, it is possible to control the computational burden of the method.
This first part of the paper focuses on the model and associated estimation algorithms while the derivation of theoretical uniqueness guarantees is postponed to the Part II article \cite{pctf3d_part2}.
Section \ref{sec:preliminary} introduces basic notions on tensors and the Canonical Polyadic Decomposition (CPD), the probability mass tensor and the link between the naive Bayes model and CPD.
Section \ref{sec:principles} presents the principles of the proposed PCTF3D together with the corresponding optimization problem and the algorithm proposed to solve it.
In Section \ref{sec:choosingT}, the core concept of PCTF3D -- the problem of choosing the 3D marginals to be coupled -- is examined through the hypergraph framework.
Section \ref{sec:numExp} features numerical experiments that provide insights on the performance of PCTF3D, especially regarding the coupling strategies presented in Section \ref{sec:choosingT}.
Finally, Section \ref{sec:appli_cyto} presents an application of PCTF3D to the problem of flow cytometry data clustering for which the PCTF3D's model proved to be very effective.

\subsection*{Notations}
Scalars are denoted as both upper and lowercase letters ($a$ or $A$).
Vectors are denoted as lowercase bold letters whereas matrices are uppercase bold letters ($\bfa$).
High-order tensors are denoted as calligraphic bold letters ($\bcalA$).
The set of integers $\{1,\ldots, M\}$ is denoted $\cpdsetp{1,M}$ in the following.
The notation $\vUn_I$ refers to a column-vector of size $I$ whose entries are all ones.
$\vectorize(.)$ denotes the column-major vectorization operator for either a tensor or a matrix.
The operator $.^\T$ designates the transposition of either a vector or a matrix.
The identity matrix of size $R$ will be denoted as $\bfI_R$.
For a vector $\bfd$ of size $I$, $\Diag(\bfd)$ denotes the diagonal matrix of size $I\times I$ whose diagonal is the vector $\bfd$.
For an order-$M$ tensor $\bcalX$ of size $I\times \cdots\times I$, the matricization on the $m$-th mode is denoted as $\bcalX_{(m)}$ and is a matrix of size $I\times I^{M-1}$.
This matrix is storing the entries of $\bcalX$ in the order of \cite{kolda_tensor_2009}.
For two matrices $\bfA\in \dsR^{I\times R}$ and $\bfB\in\dsR^{J\times R}$, the Khatri-Rao product $\bfA\khatri\bfB$ is a matrix of size $IJ\times R$ defined as the Kronecker product of columns of $\bfA$ and $\bfB$ (see \cite{kolda_tensor_2009} for more details).
The Frobenius norm is denoted as $\|.\|_F$.

\section{Preliminaries} \label{sec:preliminary}

\subsection{Tensors and canonical polyadic decomposition}
An order-$M$ tensor is an $M$-way array whose entries are defined by $M$ indices \cite{comon_tensors_2014}.
Note that a matrix, a vector and a scalar are respectively tensors of order 2, 1 and 0.
For $M\geq3$, tensors are also denoted as \emph{higher-order} tensors.
For conciseness, the term \emph{tensor} will strictly refer to a higher-order tensor in the following.
In the following, only \emph{cubic} tensors are considered, meaning that all tensor dimensions are equal to $I$.
Similarly to matrices, tensors represent multilinear operators \cite{comon_tensors_2014}.
We recall below some key definitions on tensors and their ranks.
\begin{definition} {\textbf{Rank-one tensor} -- }
    An order-$M$ tensor $\bcalX\in \dsR^{I\times \cdots \times I}$ is said \emph{rank-one} if it can be written as the outer product $\out$ of $M$ vectors $\left\{ \bfa^{(m)} \in\dsR^I \right\}_{m=1}^M$:
    \begin{equation*}
        \bcalX = \bfa^{(1)}\out\cdots\out\bfa^{(M)},
    \end{equation*}
    so that entries of $\bcalX$ are given by:
    \begin{equation*}
        x_{i_1\cdots i_M} = \prod\limits_{m=1}^M a^{(m)}_{i_m}, \quad 1\leq i_1 \; , \; \ldots \; , \; i_M \leq I.
    \end{equation*} 
\end{definition}
\begin{definition} {\textbf{Canonical Polyadic Decomposition (CPD)} -- }
    The \emph{Canonical Polyadic Decomposition} (CPD or CP decomposition) of a tensor $\bcalX\in\dsR^{I\times \cdots \times I}$ represents $\bcalX$ as a sum of $R$ rank-one tensors \cite{hitchcock_expression_1927}:
    \begin{equation}
        \label{eq:cpdDef}
        \bcalX = \sum\limits_{r=1}^R \lambda_r \amr{1}\out \cdots\out \amr{M},
    \end{equation}
    where $\amr{m}\in\dsR^{I}$ for $m\in \cpdsetp{1,M}$ and $r\in\cpdsetp{1,R}$ are the \emph{factors} of the decomposition.
    The vector $\lbd\in\dsR^{R}$ is called the \emph{loading vector}.
\end{definition}
The decomposition \eqref{eq:cpdDef} is a natural generalization of low-rank matrix decomposition to tensors, and therefore the value $R$ represents the rank of the decomposition.
Although there exists other notions of ranks for tensors (see \cite{comon_tensors_2014}), the term \emph{rank} will strictly refer to the notion of CP-rank in this article.
Figure \ref{fig:cpd3dex} depicts an example of the CPD in the $3$-dimensional case.

By grouping factors into so-called \emph{factor matrices} $\Am{m} = \begin{bmatrix}
    \am{m}{1} & \cdots & \am{m}{R}
\end{bmatrix} \in \mathbb{R}^{I\times R}$, the CPD of a tensor is compactly written as:
\begin{equation}
    \bcalX = \cpdsetp{\lbd;\;\Am{1}, \; \ldots \; , \Am{M}}.
    \label{eq:cpdMshort} 
\end{equation}
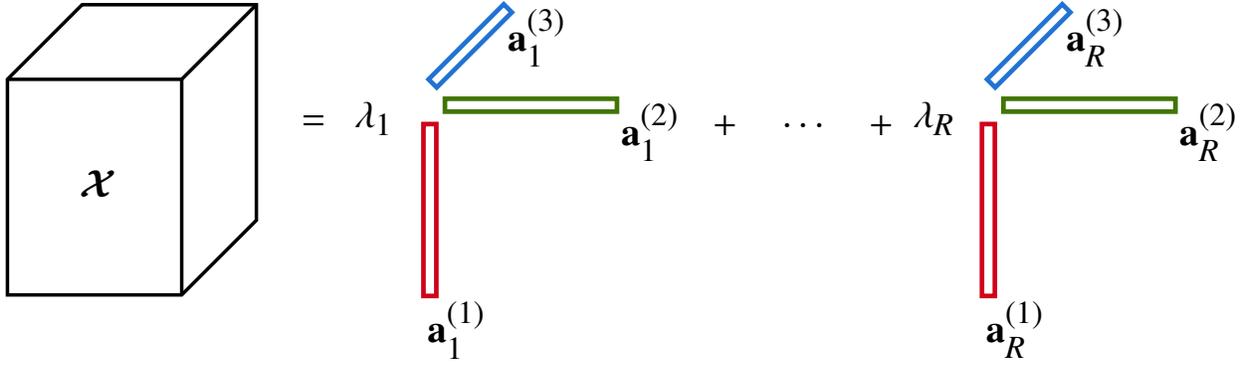
\begin{figure}
    \centerline{\resizebox{\textwidth}{!}{\tikzset{every picture/.style={line width=0.75pt}} 

\begin{tikzpicture}[x=0.75pt,y=0.75pt,yscale=-1,xscale=1]

\draw  [line width=2.25]  (44,157.76) -- (108.19,93.58) -- (257.95,93.58) -- (257.95,279) -- (193.77,343.19) -- (44,343.19) -- cycle ; \draw  [line width=2.25]  (257.95,93.58) -- (193.77,157.76) -- (44,157.76) ; \draw  [line width=2.25]  (193.77,157.76) -- (193.77,343.19) ;
\draw (121.5,247.5) node  [font=\Huge] [align=left] {$\bcalX$};

\draw  [color={rgb, 255:red, 208; green, 2; blue, 27 }  ,draw opacity=1 ][line width=3]  (401.52,196.24) -- (412.6,196.24) -- (412.6,344.08) -- (401.52,344.08) -- cycle ;
\draw (403.52,347.08) node [anchor=north west][inner sep=0.75pt]  [font=\Huge] [align=left] {$\am11$}; 
\draw  [color={rgb, 255:red, 208; green, 2; blue, 27 }  ,draw opacity=1 ][line width=3]  (881.52,196.24) -- (892.6,196.24) -- (892.6,344.08) -- (881.52,344.08) -- cycle ;
\draw (883.52,347.08) node [anchor=north west][inner sep=0.75pt]  [font=\Huge] [align=left] {$\am1R$};

\draw  [color={rgb, 255:red, 65; green, 117; blue, 5 }  ,draw opacity=1 ][line width=3]  (567.83,174.48) -- (567.83,185.56) -- (420,185.56) -- (420,174.48) -- cycle ;
\draw (569.83,177.48) node [anchor=north west][inner sep=0.75pt]  [font=\Huge] [align=left] {$\am21$};
\draw  [color={rgb, 255:red, 65; green, 117; blue, 5 }  ,draw opacity=1 ][line width=3]  (1047.83,174.48) -- (1047.83,185.56) -- (900,185.56) -- (900,174.48) -- cycle ;
\draw (1049.83,177.48) node [anchor=north west][inner sep=0.75pt]  [font=\Huge] [align=left] {$\am2R$};

\draw  [color={rgb, 255:red, 33; green, 115; blue, 211 }  ,draw opacity=1 ][line width=3]  (470.51,93.17) -- (477.49,100.15) -- (412.89,164.74) -- (405.91,157.76) -- cycle ;
\draw (472.51,96.17) node [anchor=north west][inner sep=0.75pt]  [font=\Huge] [align=left] {$\am31$};
\draw  [color={rgb, 255:red, 33; green, 115; blue, 211 }  ,draw opacity=1 ][line width=3]  (950.51,93.17) -- (957.49,100.15) -- (892.89,164.74) -- (885.91,157.76) -- cycle ;
\draw (952.51,96.17) node [anchor=north west][inner sep=0.75pt]  [font=\Huge] [align=left] {$\am3R$};

\draw (341,176) node [anchor=north west][inner sep=0.75pt]  [font=\Huge] [align=left] {$\lambda_1$};
\draw (821,176) node [anchor=north west][inner sep=0.75pt]  [font=\Huge] [align=left] {$\lambda_R$};

\draw (295,187) node [anchor=north west][inner sep=0.75pt]  [font=\Huge] [align=left] {$=$};
\draw (648.57,187) node [anchor=north west][inner sep=0.75pt]  [font=\Huge] [align=left] { $+ \quad \cdots \quad +$};

\end{tikzpicture}}}
    \caption{Example of a CP decomposition for an order-3 tensor.}
    \label{fig:cpd3dex}
\end{figure}

\subsection{Probability mass tensor}
Let $\bfz=\left(\Z{1}, \; \ldots \; ,\Z{M}\right)$ a random vector formed by $M$ random discrete variables, each $\Z{m}$ taking $I$ values from 1 to $I$\footnote{For conciseness, $I$ is considered equal for all dimensions but the method straigthforwardly generalizes to different $I_m$.}.
In this paper, we address the problem of estimating the probability mass function (PMF) of $\bfz$ denoted $\p_\bfz = \p_{\Z{1},\ldots,\Z{M}}$ from a finite set of observations.
Since the random $\bfz$ vector is discrete its PMF can be represented by an $M$-th order tensor.
\begin{definition} {\textbf{Probability mass tensor} -- }
    For a discrete random vector $\bfz$ of size $M$, the probability mass tensor $\bcalH$ of order $M$ and of size $I\times \dots \times I$ is defined as:
    \begin{equation}
        h_{\im{1}\ldots\im{M}} = \p_\bfz(i_1, \;\ldots\;, i_M) = \Pr\left(\Z{1}=\im{1}, \;\ldots\;,\Z{M}=\im{M}\right),
        \label{eq:naiveHisto}
    \end{equation}
    and contains all the information of the PMF of $\bfz$.
\end{definition}
Hence, from an observation matrix $\bfX \in \mathbb{R}^{N\times M}$ whose rows represent $N$ realizations of $\bfz$, estimating the PMF $\p_\bfz$ boils down to the estimation of the tensor $\bcalH$.
However, estimating $I^M$ values from a set of $N$ observations becomes quickly infeasible as the order $M$ grows.
This issue is known as the curse of dimensionality \cite{bellman_adaptive_1959}.

\subsection{Naive Bayes model and tensor decomposition} \label{sec:nbmPMF}
To overcome the curse of dimensionality, \cite{n_kargas_tensors_2018} proposed to model the PMF as a Naive Bayes Model (NBM) which introduces a discrete latent variable $L$ taking values from 1 to $R$ such that elements of $\bfz$ are independent conditionally to $L$:
\begin{equation}
    \label{eq:nbmpdf}
    h_{i_1\ldots i_M} = \sum\limits_{r=1}^R \Pr\left(L=r\right) \prod\limits_{m=1}^M \Pr\left(\Z{m}=i_m\mid L=r\right).
\end{equation}
Assuming an NBM for the PMF reduces the number of unknowns from $I^M$ to $M(I+1)R$ thus maintaining a complexity linear with $M$.

The NBM is a special case of graphical models \cite{jordan_graphical_2004}, which are closely related to tensor decompositions \cite{ishteva_tensors_2015,robeva_duality_2019}.
In the case of NBM, the corresponding decomposition is exactly the CPD:
\begin{align}
    h_{i_1\ldots i_M} & = \sum\limits_{r=1}^R \lambda_r \prod\limits_{m=1}^M a_{ri_m}^{(m)} \label{eq:modelCPDsum} \\
    \bcalH & = \cpdM \label{eq:CPDH},
\end{align}
where factors $\{\amr{m}\}^{R,M}_{r=1,m=1}$ contain conditional probabilities with respect to the latent variable and $\lbd$ stores the probabilities of each latent state: (i) $a^{(m)}_{ri_m}$ represents $\Pr(\Z{m}=i_m \;|\; L=r)$, and (ii) $\lambda_r$ represents $\Pr(L=r)$.
The equivalence between the tensor model \eqref{eq:CPDH} and the probabilistic model \eqref{eq:nbmpdf} results in non-negative and sum-to-one constraints on factor matrices $\Am{m}$ and $\lbd$ of the CPD:
\begin{subequations} \label{eq:simplex}
    \centerline{\begin{minipage}{0.45\linewidth}
        \begin{equation}
            \Am{m}\geq0,\quad \lbd\geq0
            \label{eq:nonneg}
        \end{equation}
    \end{minipage}
    \hfill
    \begin{minipage}{0.45\linewidth}
        \begin{equation}
            \vUn_I^{\T}\Am{m} = \vUn_R^{\T},\quad \vUn_R^{\T}\lbd=1.
            \label{eq:sumtoone}
        \end{equation}
    \end{minipage}}
\end{subequations}

\quad

\noindent In the following, the constraints \eqref{eq:nonneg} and \eqref{eq:sumtoone} will be referred to as \emph{simplex} constraints \eqref{eq:simplex}.

\begin{definition} {\textbf{3D marginal tensor} -- }
    For a subset of variables $(\Z{j},\Z{k},\Z{l})$, $\Hjkl$ denotes the 3D marginal tensor of size $I\times I\times I$ whose entries are defined by:
    \begin{equation}
        h^{(jkl)}_{i_ji_ki_\ell} = \Pr\left(\Z{j}=i_j,\Z{k}=i_k,\Z{\ell}=i_\ell\right).
        \label{eq:nbmJKLmarg}
    \end{equation}
\end{definition} 

An important property of the NBM is that it is robust to marginalization, in the sense that, if $\bcalH$ follows a NBM, any $\Hjkl$ follows an NBM with the same parameters (due to sum-to-one constraints \eqref{eq:sumtoone}):
\begin{equation}
    \bcalH = \cpdM \quad \Longrightarrow \quad \forall \jkl, \; \Hjkl = \cpdT.
    \label{eq:histo:cpdJKL}
\end{equation}
Note that in \eqref{eq:histo:cpdJKL}, the converse is true if the coupled decomposition is identifiable.
This is the topic of the joint article \cite{pctf3d_part2}.

\section{Partially Coupled Tensor Factorization of 3D marginals} \label{sec:principles}

The approach proposed by \cite{n_kargas_tensors_2018} consists in formulating the PMF estimation as a coupled tensor decomposition using all the possible empirical 3D marginals.
This will be referred to as Fully Coupled Tensor Factorization of 3D marginals (FCTF3D).
As said before, as the number of marginal increases cubically, the computational burden may rapidly become prohibitive for high dimensional problems.
This motivated the development of the Partially Coupled Tensor Factorization of 3D marginals (PCTF3D) whose principle is quite simple : instead of coupling all the marginals, it proposed to estimate the PMF using a limited number of marginals.

\subsection{Optimization problem}

Let $\calT$ be a set of triplets that will be called a \emph{coupling} in the following.
The optimization problem of PCTF3D is then the following:
\begin{align}
    & \lbdh, \; \Amh1, \; \; \ldots \; \Amh{M} = \argmin_{\lbd, \; \Am1, \; \ldots \; , \; \Am{M}} \sum\limits_{\jkl\in\calT} \left\| \Htript{j}k\ell - \cpdT \right\|^2_F \label{eq:optimPCTF3D} \\
    \text{s.t.} \quad & \lbd, \{\Am{m}\}^M_{m=1} \text{ follow \eqref{eq:simplex}.}\notag
\end{align}
In fact, the FCTF3D is a special case for which $\calT =\{i < j< k\}$.
Thorough discussions about possible couplings $\calT$ are postponed to \Cref{sec:choosingT} while the identifiability of partially coupled models is fully addressed in \cite{pctf3d_part2}.

\subsection{Algorithm}
Following the Alternating Optimization (AO) approach of \cite{k_huang_flexible_2016}, \eqref{eq:optimPCTF3D} is solved by considering sub-problems that update each factor matrix and loading vector alternatively.
The ($t+1$)-update of the $m$-th factor denoted $\Am{m}_{t+1}$ is obtained by solving the following optimization sub-problem:
\begin{align}
    & \Am{m}_{t+1} = \argmin_{\Am{m}\in\dsR^{I\times R}} \sum\limits_{ \triples{m}k\ell\in\calT} \left\| \Htript{m}k\ell_{(1)} - \Am{m}\Diag(\lbd_t)\left( \Am\ell\khatri\Am{k} \right)^\T \right\|^2_2 \label{eq:subOptimA} \\
    \text{s.t.} \quad & \Am{m} \; \text{follows \eqref{eq:simplex}.} \notag
\end{align}
Here, two remarks have to be made.
First, the update of the $m$-th factor is only based on the 3D marginals where the $m$-th variable appears.
This is what expression $\triples{m}k\ell\in\calT$ means: only the triplets of the form $\triples{m}k\ell$ that are in subset $\calT$ are considered.
Second, the terms $\Am{k}$ and $\Am\ell$ designate the last updates of those factors.
Therefore, $\Am{k}$ is referring to $\Am{k}_t$ if $k>m$ or $\Am{k}_{t+1}$ if $k<m$.
Regarding the ($t+1$)-update of $\lbd$, it is obtained by solving the following optimization problem:
\begin{align}
    & \lbd_{t+1} = \argmin_{\lbd\in\dsR^{R}} \sum\limits_{ \triples{j}k\ell\in\calT} \left\| \; \vectorize(\Htript{j}k\ell) - \left( \Am\ell_{t+1}\khatri\Am{k}_{t+1} \khatri \Am{j}_{t+1} \right)\lbd \;\right\|^2_2 \label{eq:subOptimL} \\
    \text{s.t.} \quad & \lbd \; \text{ follows \eqref{eq:simplex}.} \notag
\end{align}
Solving optimization problems \eqref{eq:subOptimA} and \eqref{eq:subOptimL} is then performed thanks to the classical ADMM \cite{boyd_distributed_2011}.
The update procedures are thoroughly detailed in \ref{app:updates}.
For both sub-problems, the stopping criterion is the same as in \cite[Eq. (10) and (11)]{k_huang_flexible_2016}.
The whole procedure is summarized in Algorithm \ref{alg:PCTF3D}.
\begin{algorithm2e}
    \SetAlgoLined
    \KwIn{$\bfX$, $\calT$, $R$.}
    \KwOut{$\{\Am{m}\}_{m=1}^M$ and $\lbd$.}
    \KwSty{Initialization:$\;$}{$\{\Am{m}\}_{m=1}^M$ and $\lbd\in \dsR^R$ random with simplex constraints \eqref{eq:nonneg} and \eqref{eq:sumtoone}.}
     
    Estimate the empirical marginals $\Hjklt$ for $\jkl\in\calT$.

    \While{Convergence is not achieved}{
        \For{$m \in\cpdsetp{1,M}$}{
            Update of $\Am{m}$ by solving optimization problem \eqref{eq:subOptimA}.
        }
        Update of $\lbd$ by solving optimization problem \eqref{eq:subOptimL}.}
    \caption{PCTF3D (solving \eqref{eq:optimPCTF3D}).}
    \label{alg:PCTF3D}
\end{algorithm2e}
\section{How to choose the marginals to be coupled?} \label{sec:choosingT}
In this section, the problem of choosing the subset of coupled 3D-marginals is presented.
Setting the so-called \emph{partial coupling} comes to the choice of $T$ different triplets from the $\binom{M}{3}$ possible triplets.
In the following, a \emph{coupling strategy} denotes a method to choose a coupling at a fixed given number of triplets $T$.
Note that couplings presented in the following subsections are defined up to a permutation of variables.
Coupling strategies are summarized in \Cref{tab:summaryCouling}.

\begin{table}
    \centerline{\begin{tabular}{ccccc}
        \toprule
        Coupling strategy & Number of triplets & Properties \\ \midrule
        Full & $\binom{M}{3}$ & Corresponds to FCTF3D \cite{n_kargas_tensors_2018} \\ \\
        '+2' & $\left\lfloor\frac{M}2\right\rfloor$ & A least coupled strategy \\ \\
        '+1' & $M$ & A least coupled strategy with constant sequence of degrees \\ \\
        Random & $T$ & Unbalanced sequence of degrees \\ \\ 
        Balanced & $T$ & Step-1 coupling \\
        \bottomrule
    \end{tabular}}
    \caption{Coupled strategies properties and associated complexities.}
    \label{tab:summaryCouling}
\end{table}
Before detailing specifically each coupling strategy, the link between hypergraphs and couplings is presented together with associated key notion \emph{sequence of degrees} from which balanced coupling is derived.

\subsection{Couplings as hypergraphs} \label{sec:couplings}
A coupling $\calT$ is a set of triplets where each triplet can be interpreted as a link between a set of 3 variables.
Thus, coupling $\calT$ is a hypergraph.
Unlike graphs that link variables by pairs, hypergraphs have vertices that can contain more than 2 variables.
\begin{definition} {\textbf{$r$-uniform hypergraph} \cite{frosini_degree_2013} -- }
    A \emph{hypergraph} is a tuple $(\calV, \calE)$ where $\calV$ is a finite set of vertices and $\calE$ is a set of \emph{hyper-edges} (a set of non-empty subsets of $\calV$).
    A hypergraph is called \emph{$r$-uniform} if all elements in $\calE$ have cardinality $r$.
\end{definition}
In the context of PCTF3D, 3-uniform hypergraphs have to be considered, meaning that each triplet $\epsilon\in\calE$ is a set containing 3 different elements of $\calV = \cpdsetp{1,M}$.
In the following, the hypergraph $\calT$ will denote the set of $T$ vertices that defines a 3-uniform hypergraph.
\begin{example}
    For example, with $M=4$ variables, the coupling of \Cref{fig:3hypergraph} is represented by the hypergraph $\calT = \{\triples{1}{2}{3},\triples{2}{3}{4}\}$.
\end{example}
Now, important notions relating to coupling as hypergraph are presented.
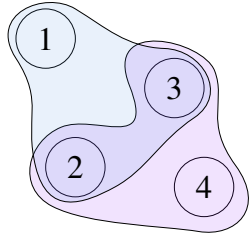
\begin{figure}
    \centerline{\resizebox{!}{3.3cm}{\tikzset{every picture/.style={line width=0.75pt}} 

\begin{tikzpicture}[x=0.75pt,y=0.75pt,yscale=-1,xscale=1]

\draw   (40,70) .. controls (40,53.43) and (53.43,40) .. (70,40) .. controls (86.57,40) and (100,53.43) .. (100,70) .. controls (100,86.57) and (86.57,100) .. (70,100) .. controls (53.43,100) and (40,86.57) .. (40,70) -- cycle ;

\draw  [color={rgb, 255:red, 0; green, 0; blue, 0 }  ,draw opacity=1 ] (70,200) .. controls (70,183.43) and (83.43,170) .. (100,170) .. controls (116.57,170) and (130,183.43) .. (130,200) .. controls (130,216.57) and (116.57,230) .. (100,230) .. controls (83.43,230) and (70,216.57) .. (70,200) -- cycle ;

\draw  [color={rgb, 255:red, 0; green, 0; blue, 0 }  ,draw opacity=1 ] (170,120) .. controls (170,103.43) and (183.43,90) .. (200,90) .. controls (216.57,90) and (230,103.43) .. (230,120) .. controls (230,136.57) and (216.57,150) .. (200,150) .. controls (183.43,150) and (170,136.57) .. (170,120) -- cycle ;

\draw  [color={rgb, 255:red, 0; green, 0; blue, 0 }  ,draw opacity=1 ] (200,220) .. controls (200,203.43) and (213.43,190) .. (230,190) .. controls (246.57,190) and (260,203.43) .. (260,220) .. controls (260,236.57) and (246.57,250) .. (230,250) .. controls (213.43,250) and (200,236.57) .. (200,220) -- cycle ;

\draw  [fill={rgb, 255:red, 144; green, 19; blue, 254 }  ,fill opacity=0.12 ] (240.65,102.56) .. controls (228.65,71.56) and (127.65,48.56) .. (154.65,117.56) .. controls (181.65,186.56) and (119.65,145.56) .. (85.65,159.56) .. controls (51.65,173.56) and (39.65,220.56) .. (76.65,239.56) .. controls (113.65,258.56) and (187.65,259.02) .. (224.65,265.02) .. controls (261.65,271.02) and (293.65,242.02) .. (261.65,198.02) .. controls (229.65,154.02) and (252.65,133.56) .. (240.65,102.56) -- cycle ;
\draw  [fill={rgb, 255:red, 74; green, 144; blue, 226 }  ,fill opacity=0.12 ] (190,80) .. controls (228.35,81.29) and (255.35,126.29) .. (220,150) .. controls (184.65,173.71) and (172.35,202.29) .. (140,220) .. controls (107.65,237.71) and (72.35,249.29) .. (60,200) .. controls (47.65,150.71) and (69.35,135.29) .. (40,90) .. controls (10.65,44.71) and (69.35,21.29) .. (100,40) .. controls (130.65,58.71) and (151.65,78.71) .. (190,80) -- cycle ;

\draw (230,220) node  [font=\huge,color={rgb, 255:red, 0; green, 0; blue, 0 }  ,opacity=1 ] [align=left] {4};
\draw (200,120) node  [font=\huge,color={rgb, 255:red, 0; green, 0; blue, 0 }  ,opacity=1 ] [align=left] {3};
\draw (100,200) node  [font=\huge,color={rgb, 255:red, 0; green, 0; blue, 0 }  ,opacity=1 ] [align=left] {2};
\draw (70,70) node  [font=\huge] [align=left] {1};

\end{tikzpicture}}}
    \caption{3-uniform hypergraph example with $M=4$ variables and 2 triplets.}
    \label{fig:3hypergraph} 
\end{figure}

\begin{definition} {\textbf{Connected hypergraph} -- }
    A hypergraph is \emph{connected} if, for any pair of vertices $\{v1,v2\}\in\calV$, there exists a path of $P$ hyper-edges $(e_1,\ldots, e_P)\in\calE^P$ such that: $v_1\in e_1$, $v_2\in e_P$ and $e_p \cap e_{p+1} \neq \emptyset$ for $1\leq p<P$.
    In the following, such couplings are denoted as \emph{valid} couplings.
    \label{def:validCoupling}
\end{definition}
To obtain an estimate of factors $\{\Am{m}\}_{m=1}^M$ from marginals $\{\Hjklt\}_{\jkl\in\calT}$, the coupling $\calT$ must be valid.
Indeed, if $\calT$ is not connected, two cases arise for which the estimation is impossible.
First, one variable does not appear in any of the triplets.
For example, \Cref{fig:couplings:miss} shows a coupling without the variable $m=6$.
Therefore, the factors $\Am6$ cannot be estimated.
Secondly, let us consider the case of two disjoints non-empty hypergraphs $(\calV_1,\calE_1)$ and $(\calV_2,\calE_2)$ such that $\calV = \calV_1 \cup \calV_2$ and $\calE = \calE_1 \cup \calE_2$.
This case results in the estimation of two uncoupled sets of factors.
For the example of \Cref{fig:couplings:dec}, one may obtain an estimation of $\{\Am1,\Am2,\Am3\}$ and $\{\Am4,\Am5,\Am6\}$.
However, the fully-joint decomposition of the 6 factors cannot be estimated.
An example of a valid coupling is given in \Cref{fig:couplings:val} where it is easy to see that: all variables are present, and any pairs of triplets may be linked by a path of triplets.
\begin{figure}
    \centering
    \begin{subfigure}[b]{0.32\linewidth}
        \centerline{\resizebox{!}{3.3cm}{\tikzset{every picture/.style={line width=0.75pt}} 

\begin{tikzpicture}[x=0.75pt,y=0.75pt,yscale=-1,xscale=1]

\draw   (40,70) .. controls (40,53.43) and (53.43,40) .. (70,40) .. controls (86.57,40) and (100,53.43) .. (100,70) .. controls (100,86.57) and (86.57,100) .. (70,100) .. controls (53.43,100) and (40,86.57) .. (40,70) -- cycle ;

\draw  [color={rgb, 255:red, 0; green, 0; blue, 0 }  ,draw opacity=1 ] (70,200) .. controls (70,183.43) and (83.43,170) .. (100,170) .. controls (116.57,170) and (130,183.43) .. (130,200) .. controls (130,216.57) and (116.57,230) .. (100,230) .. controls (83.43,230) and (70,216.57) .. (70,200) -- cycle ;

\draw  [color={rgb, 255:red, 0; green, 0; blue, 0 }  ,draw opacity=1 ] (170,120) .. controls (170,103.43) and (183.43,90) .. (200,90) .. controls (216.57,90) and (230,103.43) .. (230,120) .. controls (230,136.57) and (216.57,150) .. (200,150) .. controls (183.43,150) and (170,136.57) .. (170,120) -- cycle ;

\draw  [color={rgb, 255:red, 0; green, 0; blue, 0 }  ,draw opacity=1 ] (200,220) .. controls (200,203.43) and (213.43,190) .. (230,190) .. controls (246.57,190) and (260,203.43) .. (260,220) .. controls (260,236.57) and (246.57,250) .. (230,250) .. controls (213.43,250) and (200,236.57) .. (200,220) -- cycle ;

\draw  [color={rgb, 255:red, 0; green, 0; blue, 0 }  ,draw opacity=1 ] (270,70) .. controls (270,53.43) and (283.43,40) .. (300,40) .. controls (316.57,40) and (330,53.43) .. (330,70) .. controls (330,86.57) and (316.57,100) .. (300,100) .. controls (283.43,100) and (270,86.57) .. (270,70) -- cycle ;

\draw  [color={rgb, 255:red, 0; green, 0; blue, 0 }  ,draw opacity=1 ] (320,170) .. controls (320,153.43) and (333.43,140) .. (350,140) .. controls (366.57,140) and (380,153.43) .. (380,170) .. controls (380,186.57) and (366.57,200) .. (350,200) .. controls (333.43,200) and (320,186.57) .. (320,170) -- cycle ;

\draw  [fill={rgb, 255:red, 74; green, 144; blue, 226 }  ,fill opacity=0.12 ] (190,80) .. controls (228.35,81.29) and (255.35,126.29) .. (220,150) .. controls (184.65,173.71) and (172.35,202.29) .. (140,220) .. controls (107.65,237.71) and (72.35,249.29) .. (60,200) .. controls (47.65,150.71) and (69.35,135.29) .. (40,90) .. controls (10.65,44.71) and (69.35,21.29) .. (100,40) .. controls (130.65,58.71) and (151.65,78.71) .. (190,80) -- cycle ;
\draw  [fill={rgb, 255:red, 245; green, 166; blue, 35 }  ,fill opacity=0.12 ] (332.65,54.02) .. controls (320.65,23.02) and (241.65,18.02) .. (268.65,87.02) .. controls (295.65,156.02) and (183.65,173.02) .. (193.65,224.02) .. controls (203.65,275.02) and (243.65,254.02) .. (271.65,241.02) .. controls (299.65,228.02) and (299.65,198.02) .. (336.65,204.02) .. controls (373.65,210.02) and (405.65,181.02) .. (373.65,137.02) .. controls (341.65,93.02) and (344.65,85.02) .. (332.65,54.02) -- cycle ;
\draw  [fill={rgb, 255:red, 144; green, 19; blue, 254 }  ,fill opacity=0.12 ] (240.65,102.56) .. controls (228.65,71.56) and (127.65,48.56) .. (154.65,117.56) .. controls (181.65,186.56) and (119.65,145.56) .. (85.65,159.56) .. controls (51.65,173.56) and (39.65,220.56) .. (76.65,239.56) .. controls (113.65,258.56) and (187.65,259.02) .. (224.65,265.02) .. controls (261.65,271.02) and (293.65,242.02) .. (261.65,198.02) .. controls (229.65,154.02) and (252.65,133.56) .. (240.65,102.56) -- cycle ;

\draw (350,170) node  [font=\huge,color={rgb, 255:red, 0; green, 0; blue, 0 }  ,opacity=1 ] [align=left] {6};
\draw (300,70) node  [font=\huge,color={rgb, 255:red, 0; green, 0; blue, 0 }  ,opacity=1 ] [align=left] {5};
\draw (230,220) node  [font=\huge,color={rgb, 255:red, 0; green, 0; blue, 0 }  ,opacity=1 ] [align=left] {4};
\draw (200,120) node  [font=\huge,color={rgb, 255:red, 0; green, 0; blue, 0 }  ,opacity=1 ] [align=left] {3};
\draw (100,200) node  [font=\huge,color={rgb, 255:red, 0; green, 0; blue, 0 }  ,opacity=1 ] [align=left] {2};
\draw (70,70) node  [font=\huge] [align=left] {1};

\end{tikzpicture}}}
        \caption{Valid coupling}
        \label{fig:couplings:val}
    \end{subfigure}
    \hfill
    \begin{subfigure}[b]{0.32\linewidth}
        \centerline{\resizebox{!}{3.3cm}{\tikzset{every picture/.style={line width=0.75pt}} 

\begin{tikzpicture}[x=0.75pt,y=0.75pt,yscale=-1,xscale=1]

\draw   (20,50) .. controls (20,33.43) and (33.43,20) .. (50,20) .. controls (66.57,20) and (80,33.43) .. (80,50) .. controls (80,66.57) and (66.57,80) .. (50,80) .. controls (33.43,80) and (20,66.57) .. (20,50) -- cycle ;

\draw  [color={rgb, 255:red, 0; green, 0; blue, 0 }  ,draw opacity=1 ] (50,180) .. controls (50,163.43) and (63.43,150) .. (80,150) .. controls (96.57,150) and (110,163.43) .. (110,180) .. controls (110,196.57) and (96.57,210) .. (80,210) .. controls (63.43,210) and (50,196.57) .. (50,180) -- cycle ;

\draw  [color={rgb, 255:red, 0; green, 0; blue, 0 }  ,draw opacity=1 ] (150,100) .. controls (150,83.43) and (163.43,70) .. (180,70) .. controls (196.57,70) and (210,83.43) .. (210,100) .. controls (210,116.57) and (196.57,130) .. (180,130) .. controls (163.43,130) and (150,116.57) .. (150,100) -- cycle ;

\draw  [color={rgb, 255:red, 0; green, 0; blue, 0 }  ,draw opacity=1 ] (180,200) .. controls (180,183.43) and (193.43,170) .. (210,170) .. controls (226.57,170) and (240,183.43) .. (240,200) .. controls (240,216.57) and (226.57,230) .. (210,230) .. controls (193.43,230) and (180,216.57) .. (180,200) -- cycle ;

\draw  [color={rgb, 255:red, 0; green, 0; blue, 0 }  ,draw opacity=1 ] (250,50) .. controls (250,33.43) and (263.43,20) .. (280,20) .. controls (296.57,20) and (310,33.43) .. (310,50) .. controls (310,66.57) and (296.57,80) .. (280,80) .. controls (263.43,80) and (250,66.57) .. (250,50) -- cycle ;

\draw  [color={rgb, 255:red, 0; green, 0; blue, 0 }  ,draw opacity=1 ] (300,150) .. controls (300,133.43) and (313.43,120) .. (330,120) .. controls (346.57,120) and (360,133.43) .. (360,150) .. controls (360,166.57) and (346.57,180) .. (330,180) .. controls (313.43,180) and (300,166.57) .. (300,150) -- cycle ;

\draw  [fill={rgb, 255:red, 74; green, 144; blue, 226 }  ,fill opacity=0.12 ] (170,60) .. controls (208.35,61.29) and (235.35,106.29) .. (200,130) .. controls (164.65,153.71) and (152.35,182.29) .. (120,200) .. controls (87.65,217.71) and (52.35,229.29) .. (40,180) .. controls (27.65,130.71) and (49.35,115.29) .. (20,70) .. controls (-9.35,24.71) and (49.35,1.29) .. (80,20) .. controls (110.65,38.71) and (131.65,58.71) .. (170,60) -- cycle ;
\draw  [fill={rgb, 255:red, 245; green, 166; blue, 35 }  ,fill opacity=0.12 ] (267.65,15.71) .. controls (300.65,4.71) and (329.65,37.71) .. (306.65,75.71) .. controls (283.65,113.71) and (245.65,161.71) .. (247.65,197.71) .. controls (249.65,233.71) and (189,261) .. (176.65,211.71) .. controls (164.3,162.43) and (174.3,143.43) .. (146.65,126.71) .. controls (119,110) and (154.65,30.71) .. (200.65,61.71) .. controls (246.65,92.71) and (234.65,26.71) .. (267.65,15.71) -- cycle ;

\draw (50,50) node  [font=\huge] [align=left] {1};
\draw (80,180) node  [font=\huge,color={rgb, 255:red, 0; green, 0; blue, 0 }  ,opacity=1 ] [align=left] {2};
\draw (180,100) node  [font=\huge,color={rgb, 255:red, 0; green, 0; blue, 0 }  ,opacity=1 ] [align=left] {3};
\draw (210,200) node  [font=\huge,color={rgb, 255:red, 0; green, 0; blue, 0 }  ,opacity=1 ] [align=left] {4};
\draw (280,50) node  [font=\huge,color={rgb, 255:red, 0; green, 0; blue, 0 }  ,opacity=1 ] [align=left] {5};
\draw (330,150) node  [font=\huge,color={rgb, 255:red, 0; green, 0; blue, 0 }  ,opacity=1 ] [align=left] {6};

\end{tikzpicture}}}
        \caption{One missing variable}
        \label{fig:couplings:miss}
    \end{subfigure}
    \hfill
    \begin{subfigure}[b]{0.32\linewidth}
        \centerline{\resizebox{!}{3.3cm}{\tikzset{every picture/.style={line width=0.75pt}} 

\begin{tikzpicture}[x=0.75pt,y=0.75pt,yscale=-1,xscale=1]

\draw   (20,50) .. controls (20,33.43) and (33.43,20) .. (50,20) .. controls (66.57,20) and (80,33.43) .. (80,50) .. controls (80,66.57) and (66.57,80) .. (50,80) .. controls (33.43,80) and (20,66.57) .. (20,50) -- cycle ;

\draw  [color={rgb, 255:red, 0; green, 0; blue, 0 }  ,draw opacity=1 ] (50,180) .. controls (50,163.43) and (63.43,150) .. (80,150) .. controls (96.57,150) and (110,163.43) .. (110,180) .. controls (110,196.57) and (96.57,210) .. (80,210) .. controls (63.43,210) and (50,196.57) .. (50,180) -- cycle ;

\draw  [color={rgb, 255:red, 0; green, 0; blue, 0 }  ,draw opacity=1 ] (150,100) .. controls (150,83.43) and (163.43,70) .. (180,70) .. controls (196.57,70) and (210,83.43) .. (210,100) .. controls (210,116.57) and (196.57,130) .. (180,130) .. controls (163.43,130) and (150,116.57) .. (150,100) -- cycle ;

\draw  [color={rgb, 255:red, 0; green, 0; blue, 0 }  ,draw opacity=1 ] (180,200) .. controls (180,183.43) and (193.43,170) .. (210,170) .. controls (226.57,170) and (240,183.43) .. (240,200) .. controls (240,216.57) and (226.57,230) .. (210,230) .. controls (193.43,230) and (180,216.57) .. (180,200) -- cycle ;

\draw  [color={rgb, 255:red, 0; green, 0; blue, 0 }  ,draw opacity=1 ] (250,50) .. controls (250,33.43) and (263.43,20) .. (280,20) .. controls (296.57,20) and (310,33.43) .. (310,50) .. controls (310,66.57) and (296.57,80) .. (280,80) .. controls (263.43,80) and (250,66.57) .. (250,50) -- cycle ;

\draw  [color={rgb, 255:red, 0; green, 0; blue, 0 }  ,draw opacity=1 ] (300,150) .. controls (300,133.43) and (313.43,120) .. (330,120) .. controls (346.57,120) and (360,133.43) .. (360,150) .. controls (360,166.57) and (346.57,180) .. (330,180) .. controls (313.43,180) and (300,166.57) .. (300,150) -- cycle ;

\draw  [fill={rgb, 255:red, 74; green, 144; blue, 226 }  ,fill opacity=0.12 ] (170,60) .. controls (208.35,61.29) and (235.35,106.29) .. (200,130) .. controls (164.65,153.71) and (152.35,182.29) .. (120,200) .. controls (87.65,217.71) and (52.35,229.29) .. (40,180) .. controls (27.65,130.71) and (49.35,115.29) .. (20,70) .. controls (-9.35,24.71) and (49.35,1.29) .. (80,20) .. controls (110.65,38.71) and (131.65,58.71) .. (170,60) -- cycle ;
\draw  [fill={rgb, 255:red, 245; green, 166; blue, 35 }  ,fill opacity=0.12 ] (312.65,34.02) .. controls (300.65,3.02) and (221.65,-1.98) .. (248.65,67.02) .. controls (275.65,136.02) and (163.65,153.02) .. (173.65,204.02) .. controls (183.65,255.02) and (223.65,234.02) .. (251.65,221.02) .. controls (279.65,208.02) and (279.65,178.02) .. (316.65,184.02) .. controls (353.65,190.02) and (385.65,161.02) .. (353.65,117.02) .. controls (321.65,73.02) and (324.65,65.02) .. (312.65,34.02) -- cycle ;

\draw (50,50) node  [font=\huge] [align=left] {1};
\draw (80,180) node  [font=\huge,color={rgb, 255:red, 0; green, 0; blue, 0 }  ,opacity=1 ] [align=left] {2};
\draw (180,100) node  [font=\huge,color={rgb, 255:red, 0; green, 0; blue, 0 }  ,opacity=1 ] [align=left] {3};
\draw (210,200) node  [font=\huge,color={rgb, 255:red, 0; green, 0; blue, 0 }  ,opacity=1 ] [align=left] {4};
\draw (280,50) node  [font=\huge,color={rgb, 255:red, 0; green, 0; blue, 0 }  ,opacity=1 ] [align=left] {5};
\draw (330,150) node  [font=\huge,color={rgb, 255:red, 0; green, 0; blue, 0 }  ,opacity=1 ] [align=left] {6};

\end{tikzpicture}}}
        \caption{Two uncoupled sets of variables}
        \label{fig:couplings:dec}
    \end{subfigure}
    \caption{Examples of couplings for $M = 6$ variables}
    \label{fig:couplings}
\end{figure}
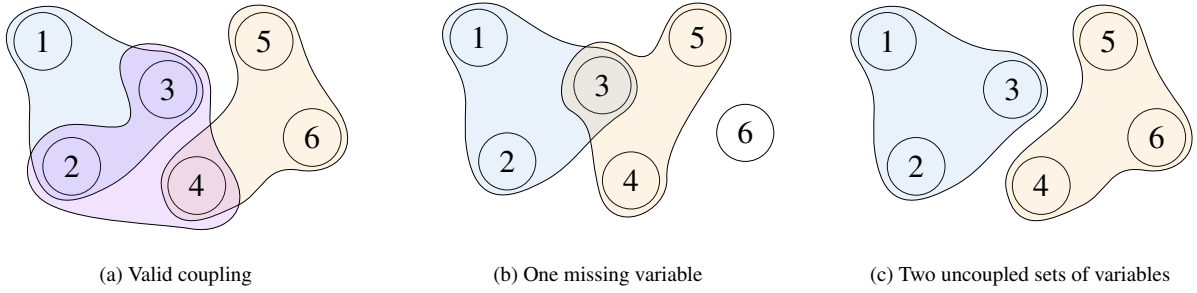

\begin{definition} {\textbf{Sequence of degrees} -- }
    For a hypergraph $\calT$, the \emph{sequence of degrees} is defined as the row vector $\bfd = \begin{bmatrix}
        d_1 & \cdots & d_M
    \end{bmatrix}\in\dsN^{1\times M}$ that contains variables occurrences inside $\calT$:
    \begin{equation*}
        d_m = \Card\left\{ \tau\in\calT \;|\; m\in\tau \right\}, \; m\in\cpdsetp{1,M}.
    \end{equation*}
\end{definition}
\begin{example}
    For the case of $\calT = \{\triples123,\triples234\}$ (see \Cref{fig:3hypergraph}), variables $1$ and $4$ are present once whereas variables 2 and 3 are represented twice.
    Then, the sequence of degrees of $\calT$ is:
    \begin{equation*}
        \bfd = \begin{bmatrix}
            1 & 2 & 2 & 1
        \end{bmatrix}.
    \end{equation*}
\end{example}
\begin{remark}
    For a valid coupling, the sequence of degrees has necessarily strictly positive elements: $d_m>0$ for all $m\in\cpdsetp{1,M}$.
\end{remark}
The sequence of degrees indicates how even the repartition of variables is across a coupling.
Intuitively, variables that appear the most are likely to be better estimated than those appearing less.
Moreover, it appears that the sequence of degrees also plays an important role in the study of the model uniqueness.
Anticipating the results of \cite{pctf3d_part2}, it appears that having $d_m=1$ in a sequence of degrees may result in defective cases, \emph{i.e.} cases where the maximal rank ensuring the recoverability of the decomposition is low.
\begin{definition} {\textbf{Step-$\alpha$ hypergraph} -- }
\label{def_step_alpha}
    For $\alpha\geq0$, a coupling $\calT$ is said \emph{step-$\alpha$} if: 
    \begin{equation*}
        (\max\limits_m d_m)-(\min\limits_m d_m) = \alpha.
    \end{equation*}
    Step-$1$ couplings are said \emph{balanced} while step-0 couplings are \emph{perfectly balanced}.
\end{definition}
For a fixed $T$, it is not always possible to find a step-$0$ coupling.
In fact, this is possible if and only if there exists an integer $\delta$ such that $3T = \delta M$ which holds if $M$ is a multiple of 3.
Note that the full coupling is always step-0, as it contains $\binom{M-1}2$ occurrences of each variable.
In \Cref{sec:balancedCoupling}, we propose an algorithm that generates step-$1$ couplings.

\subsection{A least coupling strategy: '+2' strategy}
\label{sec:leastcoupling}
First, we present a strategy with the least possible amount of triplets (in the sense of \Cref{def:validCoupling}).
The Theorem $3.9$ of \cite{boonyasombat_degree_1984} gives a necessary condition on the hypergraph connectedness which is $T\geq\lfloor M/2 \rfloor$.
Although there are several connected couplings of size $\left\lfloor\frac{M}{2}\right\rfloor$, the coupling strategy proposed here consists in adding '+2' between two triplets such that:
\begin{itemize}
    \item If $M$ is odd: \quad $\calT^{(+2)} = \{\triples123, \triples345, \ldots, \triples{M-4}{M-3}{M-2}, \triples{M-2}{M-1}M \}$,
\end{itemize}
\begin{itemize}
    \item If $M$ is even: \quad $\calT^{(+2)} = \{\triples123, \triples345, \ldots, \triples{M-3}{M-2}{M-1}, \triples{M-1}M1 \}$.
\end{itemize}
Examples of $\calT^{(+2)}$ are given in \Cref{tab:stratplus2} for small values of $M$.
\begin{table}
    \centerline{\begin{tabular}{cccc}
        \toprule
        
        $M$ & $\calT^{(+2)}$ & $\bfd$ \\ \midrule
        
        $4$ & $\{\triples123,\triples134\}$ & $\begin{bmatrix}
            2&1&2&1
        \end{bmatrix}$ \\ \\
        
        $5$ & $\{\triples123,\triples345\}$ & $\begin{bmatrix}
            1&1&2&1&1
        \end{bmatrix}$ \\ \\
        
        $6$ & $\{\triples123,\triples345,\triples156\}$ & $\begin{bmatrix}
            2&1&2&1&2&1
        \end{bmatrix}$ \\ \\
        
        $7$ & $\{\triples123,\triples345,\triples567\}$ & $\begin{bmatrix}
            1&1&2&1&2&1&1
        \end{bmatrix}$ \\ \\
        
        $8$ & $\{\triples123,\triples345,\triples567,\triples781\}$ & $\begin{bmatrix}
            2&1&2&1&2&1&2&1 
        \end{bmatrix}$ \\
        
        \bottomrule
    \end{tabular}}
    \label{tab:stratplus2}
    \caption{Examples of '+2' couplings associated with their sequence of degrees.}
\end{table}

The '+2' couplings have at most degrees equal to 2 and are by definition balanced.
Note that there exist couplings containing $T = \left\lfloor\frac{M}{2}\right\rfloor$ with more unbalanced sequences of degrees.
For example $M = 7$, $\calT = \{\triples123,\triples145,\triples167\}$ is step-($T-1$) with a sequence of degrees $\bfd = \begin{bmatrix} T & 1 & 1 & 1 & 1 & 1 & 1 \end{bmatrix}$.

\subsection{The '+1' coupling strategy}

In this section, the '+1' strategy is proposed.
This deterministic strategy generates a step-0 coupling with the least amount of triplets.
It consists in choosing $T = M$ triplets such that consecutive triplets have a 2-variables overlap.
The generated couplings are denoted as $\calT^{(+1)}$ and are obtained by adding '+1' between two consecutive triplets:
\begin{equation*}
    \calT^{(+1)} = \{\triples123,\triples234,\ldots,\triples{M-2}{M-1}M,\triples{M-1}M1,\triples{M}12\}.
\end{equation*} 
The examples featured in \Cref{tab:stratplus1} shows that '+1' couplings have constant sequence of degrees and equal to 3.
\begin{table}
    \centerline{\begin{tabular}{cccc}
        \toprule
        $M$ & $\calT^{(+1)}$ & $\bfd$ \\ \midrule
        
        $4$ & $\{\triples123,\triples234,\triples341,\triples412\}$ & $\begin{bmatrix}
            3&3&3&3
        \end{bmatrix}$ \\ \\
        
        $5$ & $\{\triples123,\triples234,\triples345,\triples451,\triples512\}$ & $\begin{bmatrix}
            3&3&3&3&3
        \end{bmatrix}$ \\ \\
        
        \multirow{2}{*}{$6$} & $\{\triples123,\triples234,\triples345,$ & \multirow{2}{*}{$\begin{bmatrix}
            3&3&3&3&3&3
        \end{bmatrix}$}\\
        & $\phantom{\{}\triples456,\triples561,\triples612\}$ & \\ \\
        
        \multirow{2}{*}{$7$} & $\{\triples123,\triples234,\triples345,\triples456,$ & \multirow{2}{*}{$\begin{bmatrix}
            3&3&3&3&3&3&3
        \end{bmatrix}$} \\
        & $\phantom{\{}\triples567,\triples671,\triples712\}$ & \\ \\
        
        \multirow{2}{*}{$8$} & $\{\triples123,\triples234,\triples345,\triples456,$ & \multirow{2}{*}{$\begin{bmatrix}
            3&3&3&3&3&3&3
        \end{bmatrix}$} \\
        & $\phantom{\{}\triples567,\triples678,\triples781,\triples812\}$ & \\
        
        \bottomrule
    \end{tabular}}
    \caption{Examples of '+1' couplings with their associated sequence of degrees.}
    \label{tab:stratplus1}
\end{table}

\subsection{Random coupling strategies}

As its name suggests, the random strategy consists in picking $T$ triplets randomly (in a uniform manner) while ensuring that the coupling is valid.
The advantage of random strategies is to control the complexity of PCTF3D, as it highly depends on the number of triplets $T$ (see \Cref{fig:BalVsRngRuntimes}).

However, random couplings have unbalanced sequence of degrees.
To apprehend this, it is possible to study $d_m$ as a random variable.
Picking $T$ triplets from the set of $\binom{M}3$ is a drawing without replacement.
And from the whole set of triplets, $\binom{M-1}2$ contains the $m$-th variable.
Therefore, at first glance, $d_m$ follows a hypergeometric distribution.
Nevertheless, the resulting set $\calT$ must be connected, which deviates $d_m$ from a hypergeometric distribution.
\Cref{fig:repart:deg} shows the empirical distributions (blue curves) of $d_m$ for different values of $T\in\cpdsetp{\lfloor\frac{M}2\rfloor,\binom{M}3}$ in the case of $M=7$ variables.
These curves were obtained after generating $1000$ random coupling yielding $7\times 1000 = 7000$ realizations of the random variable $d_m$.
They are comparable to the hypergeometric distribution (in red), with larger discrepancies arising when the number of triplets is low.
This is because hypergraph connectedness is much harder to achieve for low values of triplets.

Moreover, knowing the distribution of $d_m$ gives insights on the recoverability bounds for a value of $T$.
In particular, it appears that $d_m=1$ is a critical value which may result in defective cases, \emph{i.e.} cases where the maximal rank for having recoverability is (unexpectedly) low.
For $M=7$, when $T>21$, then $d_m>1$ empirically which leads to say that these defective cases do not appear anymore.
This will be thoroughly studied in \cite{pctf3d_part2}.
\begin{figure}
    \centering 
    \centerline{\includegraphics[height=8cm, clip, trim={50 20 70 60}]{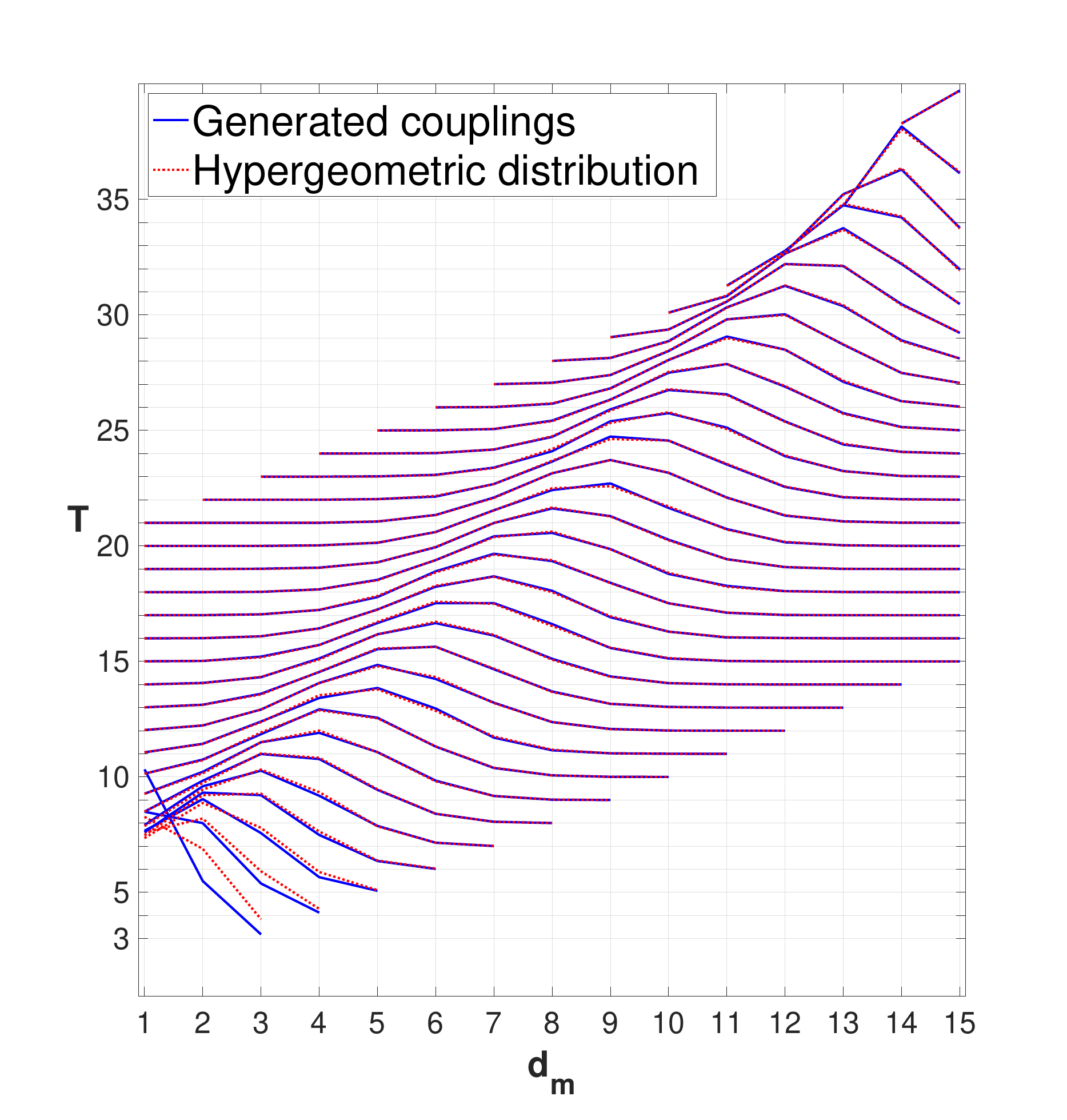}}
    \caption{ \centering Both empirical and theoretical distribution of $d_m$ for random couplings and different values of $T$ ($M=7$).}
    \label{fig:repart:deg}
\end{figure}

\subsection{Balanced couplings generation}
\label{sec:balancedCoupling}

Balanced couplings correspond to step-0 or step-1 couplings (see Definition \ref{def_step_alpha}), ensuring that all variables are represented evenly in terms of occurrences.
Such couplings are generated with the Algorithm \ref{alg:balancedcoupling} which is adapted from \cite[Section 4]{frosini_degree_2013}.
It relies on the use of Lyndon words to generate perfectly balanced couplings (hypergraphs), and can be seen as an extension of the "+1" strategy.
The reason for developing a balanced coupling strategy was motivated by two reasons.
First, numerical experiments in the next section show that balanced couplings lead to better estimations of the PMF, compared to random couplings.
Second, theoretical considerations that will be discussed in the second paper.
Indeed, unbalanced coupling may result in situations where the maximum rank for having recoverability is much lower than the one reached by other more favorable coupling.
This results in the so-called defective cases.
Adopting, a balanced coupling prevent the occurrence of such defectives cases.

Before introducing Lyndon words, recall that a triplet of variable $\{j,k,\ell\}\subset\cpdsetp{1,M}$ can also be represented as a sequence of $M$ symbols ('$0$' and '$1$').
For example, for 6 variables, the triplet $\{2,4,5\}$ may be represented as '$010110$' in the following.
\begin{definition} {\textbf{Lyndon word (LW)} -- }
    A \emph{Lyndon word} (LW) is a word of length $M$ denoted $w_1$ such that any non-trivial circular permutation $w_2$ of $w_1$ is strictly greater regarding the lexicographic order \cite{harzheim_ordered_2005} ($w_2>w_1$).
    In our framework, a LW is a binary string containing exactly 3 times the symbol '$1$'.
\end{definition}

For example, for $M=6$, there are 3 LWs which are $'000111'$, '$001011$' and '$001101$'.

Note that the strict inequality in the LW definition does not allow for periodic patterns.
This means that for $M =6$, the word '$010101$' is not considered a LW.

\subsubsection{Generation of all possible triplets from Lyndon words} 
From one LW and its circular permutations, it is possible to generate a coupling $\calT$ which is a perfectly-balanced (or step-0) with a sequence of degrees equal to 3.
\begin{example}
    Consider the LW $w = \text{'}000111\text{'}$ of length $M=6$; the hypergraph $\calT$ below has a step-0 sequence of degrees equal to 3:
    \begin{align}
            & \begin{array}{cccccc}
                    \text{'}\bm{0} & \bm{0} & \bm{0} & \bm{1} & \bm{1} & \bm{1}\text{'} \\
                    \text{'}1 & 0 & 0 & 0 & 1 & 1\text{'} \\
                    \text{'}1 & 1 & 0 & 0 & 0 & 1\text{'} \\
                    \text{'}1 & 1 & 1 & 0 & 0 & 0\text{'} \\
                    \text{'}0 & 1 & 1 & 1 & 0 & 0\text{'} \\
                    \text{'}0 & 0 & 1 & 1 & 1 & 0\text{'}
            \end{array} \quad \begin{array}{cc}
                    \calT = & \{ \triples456, \\
                    & \phantom{\{} \triples156, \\
                    & \phantom{\{} \triples126, \\
                    & \phantom{\{} \triples123, \\
                    & \phantom{\{} \triples234, \\
                    & \phantom{\{} \triples345 \} \\
            \end{array}  \\
            \bfd = & \begin{bmatrix}
                    \phantom{\text{'}}3 & 3 & 3 & 3 & 3 & 3\phantom{\text{'}}
            \end{bmatrix}.
    \end{align}
\end{example}
We now prove a property that is at the heart of the proposed algorithm for balanced coupling.
\begin{lemma}
    If two LW $(w_1,w_2)$ share a circular permutation $w$, then $w_1=w_2$.
    \label{lemma:sharingPerm}
\end{lemma}
\begin{proof}
    If $w$ is a circular permutation of both LWs, there exists $p_1$ and $p_2$ two circular permutations such that:
    \begin{equation*}
        w = p_1(w_1) \quad \text{and} \quad w = p_2(w_2).
    \end{equation*} Hence, $p_1(w_1)=p_2(w_2)$ so there exists a circular permutation $p$ such that $w_1=p(w_2)$ and $w_2=p^{-1}(w_2)$.
    If $p$ is the identical permutation the proof is complete.
    If not, it means that $w_1$ is a circular permutation of the LW $lw_2$ therefore $w_1>lw_2$ and conversely $w_2$ is a circular permutation of the LW $w_1$ so $w_2>w_1$ which leads to a contradiction.
    Hence, $p$ is the identical permutation and the proof is complete.
\end{proof} 
The converse of Lemma \ref{lemma:sharingPerm} states that for two different LWs, the 2 sets of $M$ triplets generated by their respective circular permutations are disjoints.
For our algorithm, this property permits to guarantee that all possible triplets may be generated with LWs.

In \cite{gilbert_symmetry_1961}, the number of LWs is studied in the general case of LWs.
In our case of triplets, for a given $M\geq 3$, the number of LWs is denoted $\LW(M)$ and is given by:
\begin{equation*}
    \LW(M) = \left\{ \begin{array}{lll}
        \frac{1}{M} \binom{M}3 - \frac{1}{M} \binom{\frac{M}3}1 & \text{if} & 3\mid M \\
        \frac{1}{M} \binom{M}3 & \text{else} &
    \end{array}\right.\quad \Longrightarrow \quad \LW(M) = \left\lfloor\frac{\binom{M}3}{M}\right\rfloor.
\end{equation*}

Therefore, two cases emerge.
First, if $3\nmid M$, by invoking the converse of \Cref{lemma:sharingPerm}, LWs generate $M\binom{M}3/M$ triplets, hence all possible triplets.
However, in the case of $M$ is divided by 3, there are $M/3$ triplets that are missing from the LWs circular permutations.
It comes from the fact that Lyndon words cannot be periodic, which can only occur for $3\mid M$.
For example, the triplets '$010101$' and '$101010$' are neither LWs nor a LW permutation by definition.
Therefore, the $M/3$ remaining triplets are defined as the circular permutations of:
\begin{equation}
    \text{'}\overbrace{0\cdots01}^{\frac{M}{3}}\;\overbrace{0\cdots01}^{\frac{M}{3}}\;\overbrace{0\cdots01}^{\frac{M}{3}}\text{'}.
    \label{eq:mlmdiv3}
\end{equation}
It is now possible to build an algorithmic procedure to generate balanced couplings based on LWs.

\subsubsection{An algorithm for balanced coupling generation}

\textbf{Algorithm \ref{alg:balancedcoupling}} is adapted from \cite[Section 4]{frosini_degree_2013}; it generates a balanced coupling for given values of $M$ and $T$ such as $T \leq \binom{M}{3}$.
Indeed, writing :
$$
T = \left\lfloor \frac{T}{M} \right\rfloor M + \epsilon
$$
with $\epsilon < M$.
The algorithm relies on circular permutations of $\left\lfloor \frac{T}{M} \right\rfloor$ LWs (and the word \eqref{eq:mlmdiv3} if M is divided by 3), excluding the LW $\text{'}(0)^{M-3}(1)^3\text{'}$ which is reserved for the $\epsilon$ missing triplets built from permutation of the LWs $\text{'}(0)^{M-3}(1)^3\text{'}$.
Note that when $M/2 \leq T <M$, a had-hoc deterministic procedure derived from the '+2' strategy is used.

\begin{algorithm2e}
        \KwIn{$M$ and $T$.}
        
        \If{$T\geq M$}{
                Add the circular permutations of $\lfloor T/M\rfloor$ LWs to $\calT$,
                
                If $3\mid M$ and $\Card(\calT)\neq T$, add permutations of the word \eqref{eq:mlmdiv3} to $\calT$ (as much as needed and at most $M/3$).

                If $\Card(\calT)\neq T$, add circular permutations of '$(0)^{M-3}(1)^3$' (as much as needed but at most $M$),
            }
        
        \If{$T< M$}{
                Initialize $\calT = \calT^{(+2)}$,

                Add the remaining $T-\lfloor M/2\rfloor$ triplets so that $\calT$ is step-1.
            }

        \KwOut{$\calT$}
        \caption{Balanced hypergraph generation}
        \label{alg:balancedcoupling}
\end{algorithm2e}

\subsubsection{Examples of balanced couplings generation}

\begin{example} {\textbf{$M = 7$ and $T = 12$} -- }
    For this case, one LW is needed: for example '$0001011$'.
    After adding all permutations of this word, the sequence of degree is temporarily constant and equal to 3.
    Then, $\calT$ lacks 5 triplets which will be generated from permutations of '$0000111$':
    \begin{align*}
                & \begin{array}{ccccccc}
                        \text{'}\bm{0} & \bm{0} & \bm{0} & \bm{0} & \bm{1} & \bm{1} & \bm{1}\text{'} \\
                        \text{'}0 & 1 & 1 & 1 & 0 & 0 & 0\text{'} \\
                        \text{'}1 & 0 & 0 & 0 & 0 & 1 & 1\text{'} \\
                        \text{'}0 & 0 & 1 & 1 & 1 & 0 & 0\text{'} \\
                        \text{'}1 & 1 & 0 & 0 & 0 & 0 & 1\text{'}
                \end{array} \quad \begin{array}{cc}
                        \calT = \calT \quad \cup & \{ \triples567, \\
                        & \phantom{\{} \triples234, \\
                        & \phantom{\{} \triples167, \\
                        & \phantom{\{} \triples345, \\
                        & \phantom{\{} \triples127 \} \\
                \end{array} \\
                \bfd = \bfd + & \begin{bmatrix}
                        \phantom{\text{'}}2 & 2 & 2 & 2 & 2 & 2 & 3\phantom{\text{'}}
                \end{bmatrix}.
    \end{align*}
    Finally, the resulting coupling is 
    \begin{align*}
        \calT = \{ & \triples467,\;\triples157,\;\triples126,\;\triples237,\;\triples134,\;\triples245,\;\triples356, \\
        & \triples567,\;\triples234,\;\triples167,\;\triples345,\;\triples127 \},
    \end{align*}
    which gives a sequence of degrees equal to $\bfd = \begin{bmatrix}
        5 & 5 & 5 & 5 & 5 & 5 & 6
    \end{bmatrix}$ hence $\calT$ is step-1.
    
    Note that the permutations considered in the remaining triplets are chosen such that the coupling stays step-1, which is clearly feasible in all possible case.
\end{example}

\begin{example} {\textbf{$M=6$ and $T=9$} -- }
    Similarly, the circular permutations of '$001101$' are added to $\calT$, hence it lacks 3 triplets.
    After adding the $M/3 = 2$ permutations of '$010101$' (itself and '$101010$'), only one triplet is missing.
    Finally, '$000111$' is added to give: 
    \begin{align*}
        \calT = \{ \triples346,\;\triples145,\;\triples256,\;\triples136,\;\triples124,\;\triples235,\;\triples246,\;\triples135,\;\triples456\},
    \end{align*}
    with $\bfd = \begin{bmatrix}
        4 & 4 & 4 & 5 & 5 & 5
    \end{bmatrix}$.
\end{example}

\begin{example} {\textbf{$M=10$ and $T=8$}}
    To ensure that the coupling is valid, $\calT$ is initialized with the '+2' coupling (see \Cref{sec:couplings}):
    \begin{equation*}
        \calT^{(+2)} = \{\triples123,\; \triples345,\; \triples567,\; \triples789,\; \triples19{10}\}.
    \end{equation*}
    Then, the 3 triplets missing are added so that the coupling is at least step-1 which results in:
    \begin{equation*}
        \calT = \{\triples123,\; \triples345,\; \triples567,\; \triples789,\; \triples19{10},\; \triples246,\; \triples8{10}1,\; \triples234\},
    \end{equation*}
    which admits $\bfd = \begin{bmatrix}
        3 & 3 & 3 & 3 & 2 & 2 & 2 & 2 & 2 & 2 
    \end{bmatrix}$ as its sequence of degrees.
\end{example}
\section{Numerical experiments} \label{sec:numExp}

This section presents numerical experiments that were performed to study the performance of PCTF3D (and associated coupling strategy) to that of FCTF3D which in some respects represents the best achievable performance.

\subsection{Performance criteria}
To study the performance of the proposed methods, several criteria are proposed.
First, it must be noted that the Frobenius norm between two tensors is often impossible to compute in practice, because of the prohibitive size of tensors:
\begin{equation*}
    \Err{\text{MD}} = \left\| \bcalH-\cpdsetp{\widehat{\lbd}; \Amh1, \ldots, \Amh{M}} \right\|^2_F.
\end{equation*}

We rather consider two lower-order metrics consisting in summing the errors over lower-order marginals:
\begin{equation*}
    \Err{\text{1D}} = \sum\limits_{m=1}^M \left\|\bfh^{(m)}-\widehat{\bfh}^{(m)}\right\|^2_2,
\end{equation*}
\begin{equation*}
    \Err{\text{3D}} = \sum\limits_{\jkl\in\calT_{\text{all}}} \left\| \Hjkl-\cpdsetp{\widehat{\lbd}; \Amh{j},\Amh{k},\Amh{\ell}} \right\|^2_F,
\end{equation*}
where $\bfh^{(m)}$ represents the $m$-th 1D marginal and $\widehat{\bfh}^{(m)}$ its estimation.
Note that $\Err{\text{3D}}$ is computed over all possible triplets.
It is equivalent with the cost function of the fully coupled case denoted as FCTF3D.
For PCTF3D, the 3D marginals distance is computed even for marginals not considered in the couplings.

The Factor Match Score, defined in \cite{acar_scalable_2011}, consists in finding the best permutation such that for $r\in\cpdsetp{1,R}$ estimated factors $\amrh{m}$ can be compared to the theoretical $\amr{m}$:
\begin{equation}
    \dist_{\text{FMS}}(\bcalH,\widehat{\bcalH}) = \sum\limits_{r=1}^R \prod\limits_{m=1}^M \frac{ \bfa^{(m)\T}_r\amrh{m} }{ \|\amr{m}\|_2 \|\amrh{m}\|_2 }.
    \label{eq:fms}
\end{equation}
\subsection{Influence of the number of triplets and number of samples} 

Multivariate Gaussian distributions of rank $R=20$ are randomly generated and added with weights $\lbd$.
Then, they are discretized over $I=30$ bins per dimension to create a theoretical probability mass tensor $\bcalH$.
Then, different number of samples $N\in \left\{10^4,3\cdot10^4,10^5,3\cdot10^5,10^6\right\}$ following the GMM are generated.
The 3D marginals are computed with $I=30$ bins per dimension.
We set a maximum of $10^3$ outer iterations for the stopping criterion of Algorithm \ref{alg:PCTF3D}.
A maximum of $20$ inner iterations for Algorithms \ref{alg:upA} and \ref{alg:upL}.
The target decomposition rank is set equal to the number of mixed distributions $R=20$.
The reconstruction error $\Err{\text{1D}}$ and $\Err{\text{3D}}$ are then evaluated for each strategy and averaged over 10 experiments.
\Cref{fig:ntriplets} shows that the coupling strategy '1/8' -- where $T = \frac18\binom{M}3$ triplets are chosen randomly -- yields to performance similar to the fully coupled strategy.
This shows that partial coupling strategies are beneficial in terms of computational cost as the computational complexity is linear with the number of triplets in consideration.
\begin{table}
    \label{tab:triples}
    \centering
    \small
    \begin{tabular}{ccc}
        \toprule
        Strategy & $T$ & Triplets \\ \midrule
        +2 & 5 & $\triples123$, $\triples345$, $\triples567$, $\triples789$, $\triples9{10}1$ \\ \\
        +1 & 10 & $\triples123$, $\triples234$, $\ldots$, $\triples9{10}1$, $\triples{10}12$ \\ \\
        1/8 & $15 = \frac{120}{8}$ & random triplets \\ \\
        1/4 & $30 = \frac{120}{4}$ & random triplets \\ \\
        1/2 & $60 = \frac{120}{2}$ & random triplets \\ \\
        1 & $120 = \binom{10}{3}$ & all triplets ($\sim$ FCTF3D)\\ \bottomrule
    \end{tabular}
    \caption{Coupling strategies for $M=10$.}
\end{table}
\begin{figure}
    \centering
    \begin{subfigure}{0.45\linewidth}
        \centering
        \centerline{\centerline{\includegraphics[width=\linewidth]{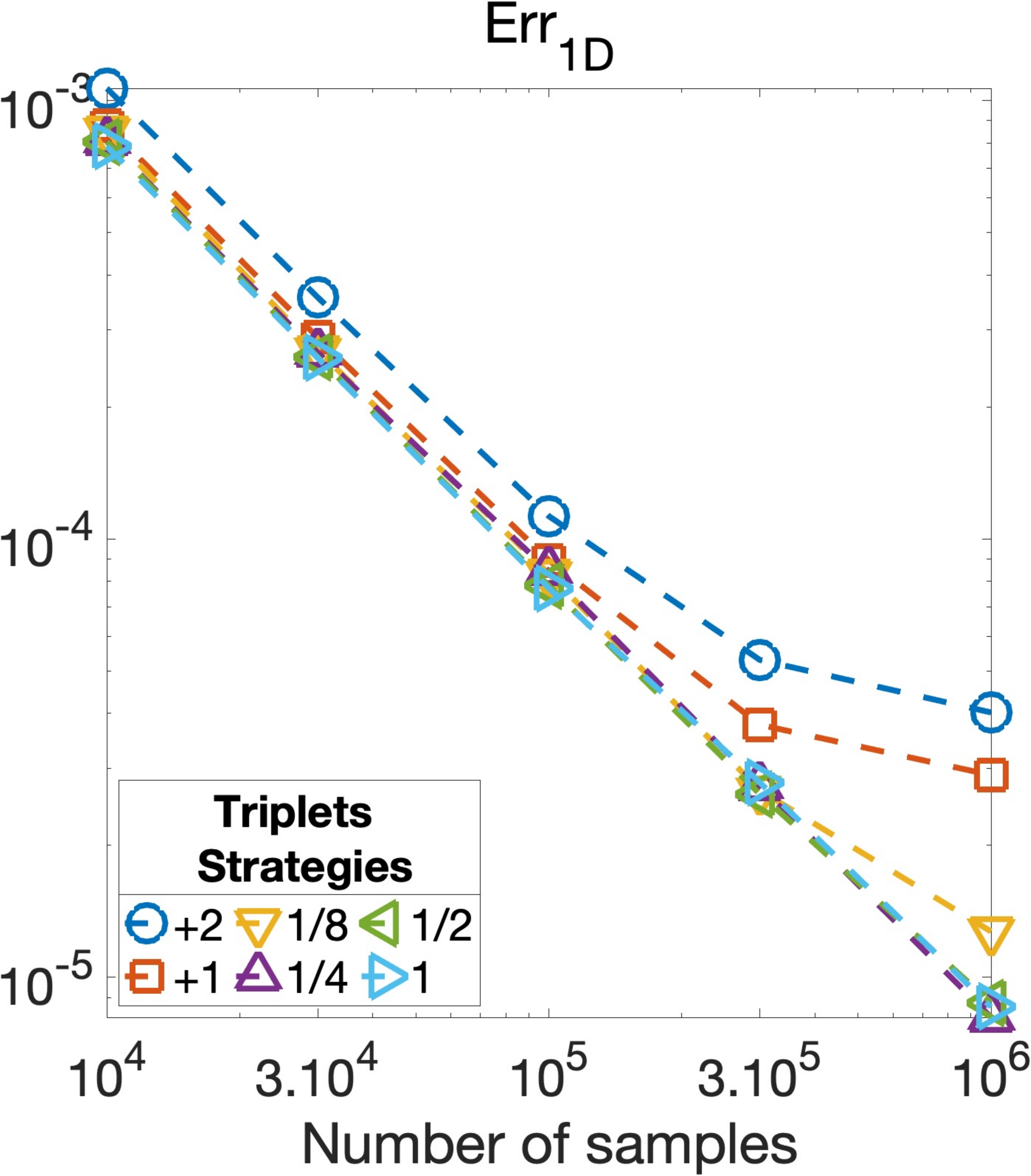}}}
        \caption{Error on 1D marginal estimation.}
        \label{eq:errmarg1d}
    \end{subfigure}
    \hfill
    \begin{subfigure}{0.45\linewidth}
        \centering
        \centerline{\centerline{\includegraphics[width=\linewidth]{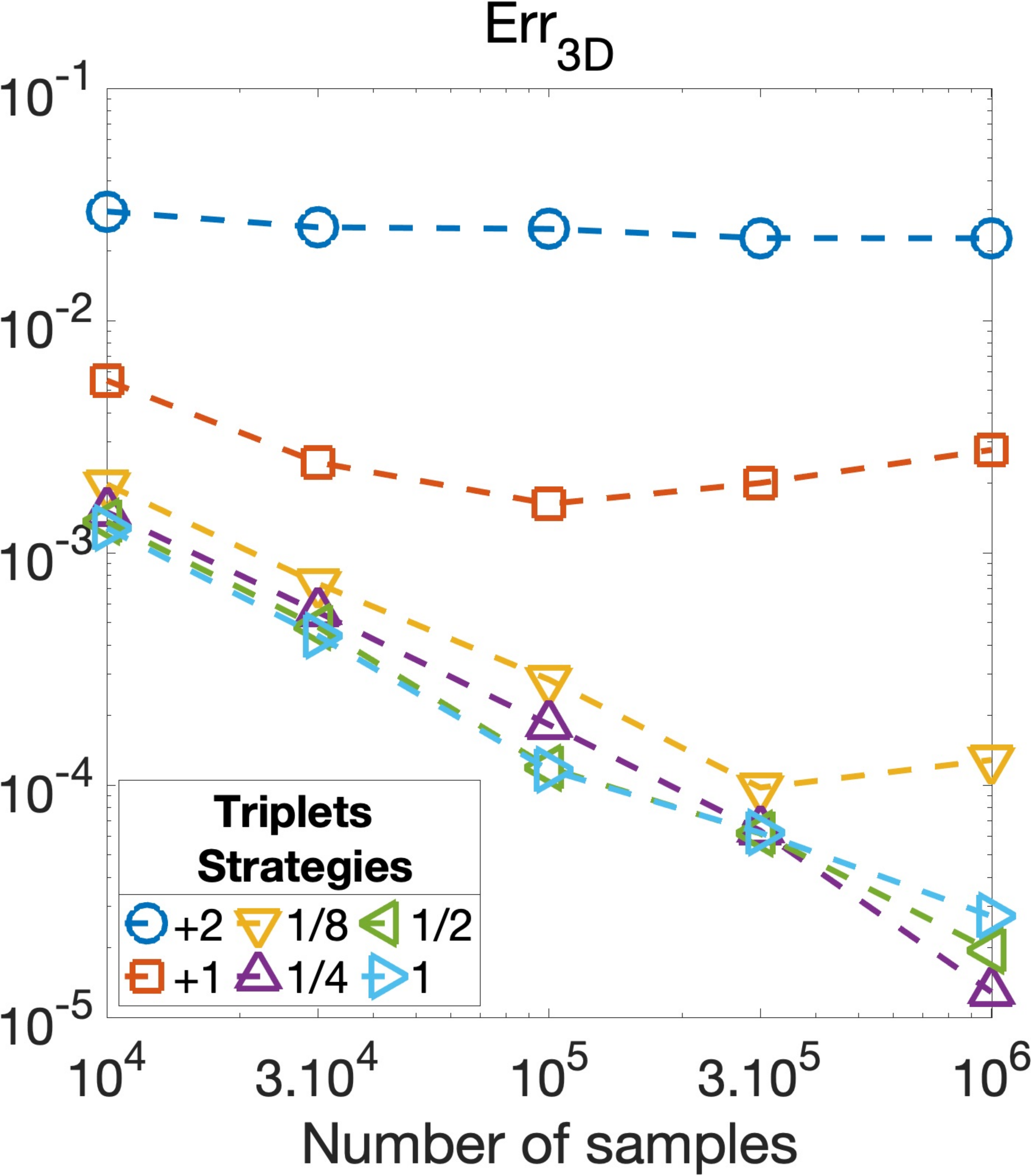}}}
        \caption{Error on 3D marginal estimation.}
        \label{eq:errmarg3d}
    \end{subfigure}
    \caption{Error regarding the number of triplets for different coupling strategies.}
    \label{fig:ntriplets}
\end{figure}

\subsection{Balanced v/s random couplings} \label{sec:rngvsbal}

To compare the random and balanced strategies, we performed a similar experiment than in the previous subsection.
However, this time uniform distributions are used.
A mixture of $R =20$ multivariate uniform distributions of dimension $M=10$ are randomly generated and discretized with $I = 15$ bins per dimension.
Then, datasets following these distributions were generated with different sample sizes:
$
    N\in\{10^4, 2.10^4, 5.10^4, 10^5, 2.10^5, 5.10^5,	10^6\}.
$
For both strategies and for each distribution, 100 couplings were generated for different number of triplets between $T=\lfloor M/2\rfloor = 5$ and $T = \binom{M}3 = 120$.

PCTF3D was then applied with $R=20$, a maximum $10^3$ outer iterations and $15$ maximum inner iterations of ADMM for each sub-problem.
The results of this experiment are shown in \Cref{fig:BalVsRng}.
Error on 3D-marginals reaches a constant level, which represents FCTF3D's performance, above a certain number of triplets.
The higher number of samples $N$, this level is achieved for fewer triplets.
First, more samples lead to a better estimation of lower-order marginals, hence reducing the final estimation errors.
Secondly, having more samples allows to reduce the number of triplets considered without losing in estimation performance.
FCTF3D's performance is achieved for both random and balanced strategies at a similar value of $T$.
The difference between those two strategies rises for lower values of $T$.
Indeed, the estimation of balanced couplings is slightly better, for both the 3D-marginals error and the FMS score.
As we will show in the second article, identifiability is stronger for balanced couplings which is a second argument in favor of balanced couplings.
Concerning PCTF3D's runtimes and complexity, Figure \ref{fig:BalVsRngRuntimes} shows runtimes for this experiment.
These computation times evolve approximately linearly with $T$ while there is nearly no complexity regarding the number of samples $N$.
Indeed, we recall that the parameter $N$ does not affect PCTF3D's computation time after the 3D marginals are computed.
More details on runtime and complexity may be found in \cite{flores_probability_2023}.
\begin{figure}
    \centerline{\includegraphics[width=.9\linewidth,clip,trim={150 45 165 10}]{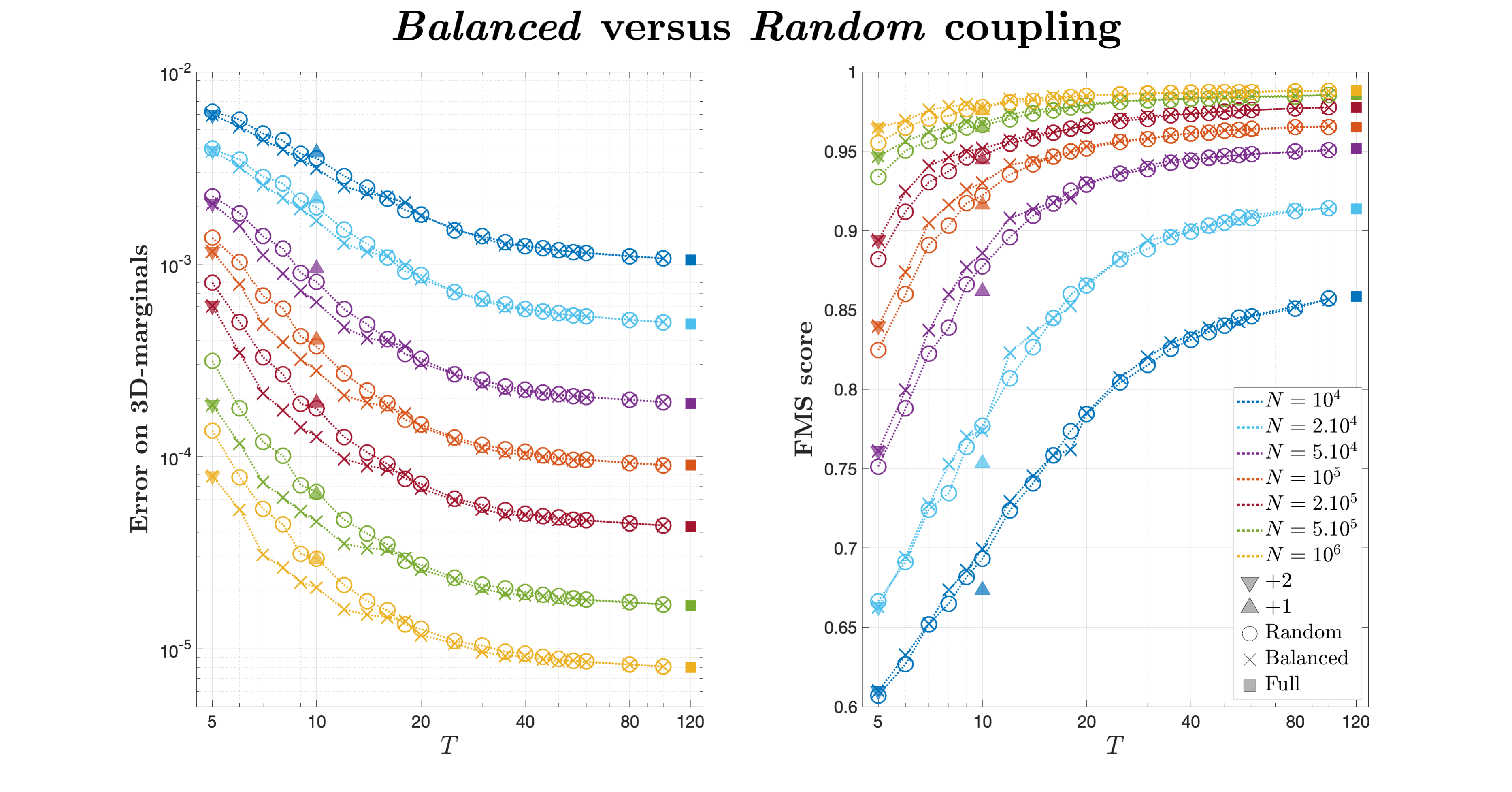}}
    \caption{Estimation performance comparison between balanced and random coupling strategies for a $10$-dimensional PMF :
    \textbf{Left}: Estimation error on 3D-marginals regarding the number of triplets $T\in[5,120]$ and the number of samples $N$ ; 
    \textbf{Right}: FMS score \eqref{eq:fms} regarding $T$ and the number of samples $N$.
    For higher number of $N$, low values of $T$ provide similar performances compared to FCTF3D ($T = 120$).
    Balanced couplings have slightly better performances compared to random couplings at low values of $T$.}
    \label{fig:BalVsRng}
\end{figure}
\begin{figure}
    \centerline{\includegraphics[width=.5\linewidth,clip,trim={40 10 80 25}]{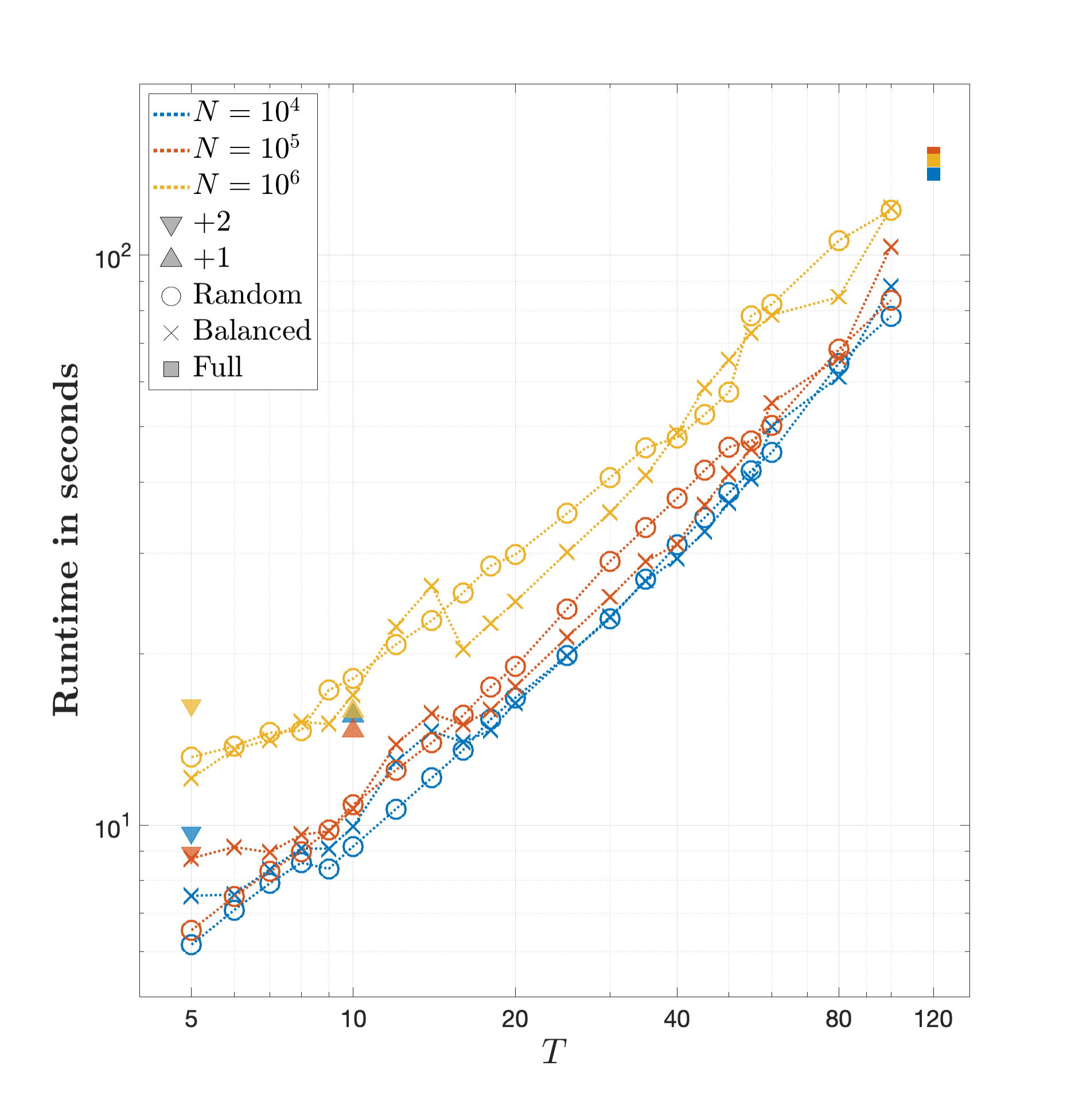}}
    \caption{Runtimes regarding the number of triplets $T$ for all coupling strategies and for different number of samples $N\in\{10^4,10^5,10^6\}$ in a case of a 10-dimensional PMF estimation problem.}
    \label{fig:BalVsRngRuntimes}
\end{figure}

\section{Application to flow cytometry data analysis} \label{sec:appli_cyto}

Flow cytometry (FCM) is a single cell analysis technique that is used in many areas including environmental studies \cite{vesey_detection_1994}, agronomy \cite{seidel_jr_sexing_2003} and oncology \cite{greve_flow_2012}.
More particularly, it has become the reference method in immunology allowing to assess the response of immune blood cells for different pathologies and their treatments \cite{daveni_myeloid-derived_2020,chattopadhyay_good_2010}.
The principle of FCM is to measure $M$ biological properties, also called fluorescences, for large number of cells $N$.
For each cell, the fluorescence values are related to its biological properties, hence permitting a characterization of cell populations regarding biological properties of interests.
In the past decades, the development of high-throughput cytometry made possible to measure more biological properties (up to 40) for more cells (up to millions) \cite{perfetto_seventeen-colour_2004,park2020omip}.

The data acquisition results in a matrix $\bfX\in\dsR^{N\times M}$ containing the $M$ biological properties of the $N$ cells and the goal of FCM data analysis is to sort and quantify cell populations of interests.
The standard method for analyzing FCM data is the manual\emph{gating} which consists in selecting cells population from bivariate point clouds until a cell population of interest is found.
However, as mentioned in \cite{lahnemann_eleven_2020}, the \emph{gating} hardly scales to large values of $M$ as it becomes more subjective and time-consuming, thus motivating the development of computational FCM.

This section aims at showing how the proposed PCTF3D can be used in the processing of FCM data.
First the processing pipeline is presented.
Then an application to the analysis of an 8-colors dataset allows to illustrate the interest of the approach.

\subsection{Processing pipeline}

The processing pipeline consists in 4 main steps :
\begin{enumerate}
    \item Choosing the number of triplets and the coupling strategy.
    \item Computing the corresponding 3D histograms.
    The choice of the number of bins per dimension should be made in accordance with the number of available sample.
    \item Estimating the $R$ factors of the NBM using PCTF3D.
    In practice, $R$ is chosen much greater than the expected number of cell populations.
    \item Performing the clustering by using the estimated factors and visualize the results.
\end{enumerate}
At this point, some comments need to be made.
First, it may be difficult to give a biological interpretation directly from the $R$ rank-1 components.
The goal of step 4 is to group the rank-1 components to obtain biologically interpretable populations.
Another interesting feature of the NBM is that it allowed to develop original and interpretable visualization tools.
All the tools are integrated in computational framework available at \url{https://github.com/philippeflores/fcm_ctflowhd}.

\subsection{Application to 8-marker datasets} \label{sec:cytoAppli:8d}
 
In this section, PCTF3D is applied to a graft sample used in the treatment of leukemia.
After pre-processing steps, the FCM dataset acquired includes $N = 1,509,790$ cells.
For each cell, $M=8$ fluorescence signal corresponding to the following markers: CD11b, CD3, PD-L1, CD33, CD14, CD34, CD15, HLA-DR, are measured.
The following experimental set-up is considered: 
\begin{itemize}
    \item Full coupling ($T = 56$ triplets) and balanced coupling ($T = 14$, '1/4' of all triplets) strategies are considered 
    \item $I=30$ bins per dimension,
    \item $R = 100$ rank-one terms,
\end{itemize}
Complete linkage hierarchical clustering \cite{milligan_ultrametric_1979} is applied to the results of PCTF3D.
This dataset was also analyzed and interpreted with manual \emph{gating} allowing to identify cell populations.

The results of \Cref{fig:cytoAppli:g129hierarc} shows the clustering of the factors obtained with the full coupling strategy.
Note that only $R=91$ components are considered, since the remaining 9 factors had a $\lambda_r$ value numerically equal to 0.
\Cref{fig:cytoAppli:g129hiercT4} shows the results for partially coupled tensor factorization.
Even if the components are not exactly the same when compared to full coupling, the resulting proportions and fluorescence properties are very similar to the one obtained in \Cref{fig:cytoAppli:g129hierarc}.
While FCTF3D ran in 22 minutes, PCTFD3's runtime was around 5 minutes, a factor 4 reduction of the computation time, the same ratio between the number of triplets of both strategies.
\begin{figure}
    \centerline{\includegraphics[width=\linewidth,clip,trim={170 35 160 40}]{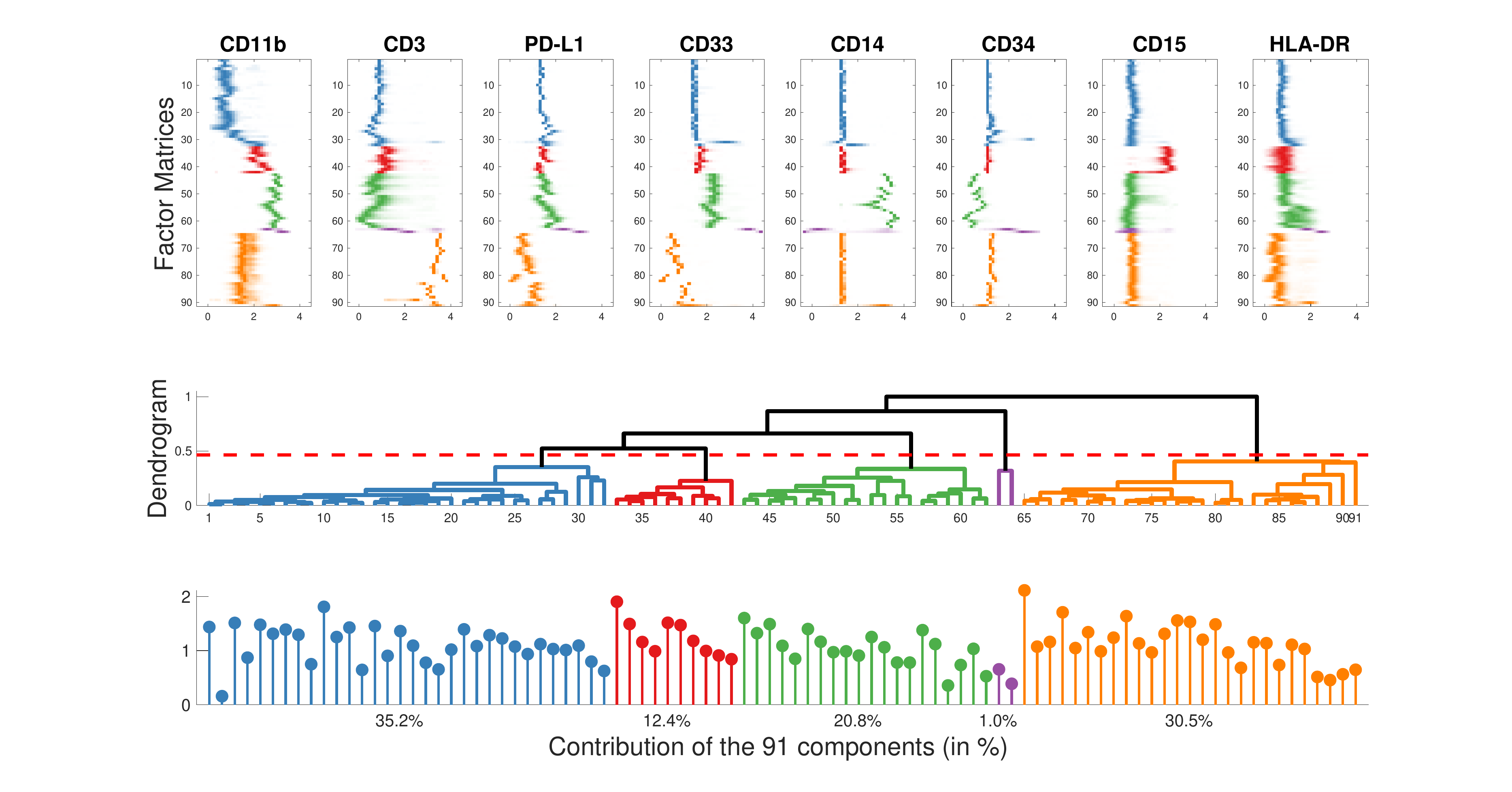}}
    \caption{PCTF3D applied to a graft sample dataset with 8 markers (full coupling).}
    \label{fig:cytoAppli:g129hierarc}
\end{figure}
\begin{figure}
    \centerline{\includegraphics[width=\linewidth,clip,trim={170 35 160 40}]{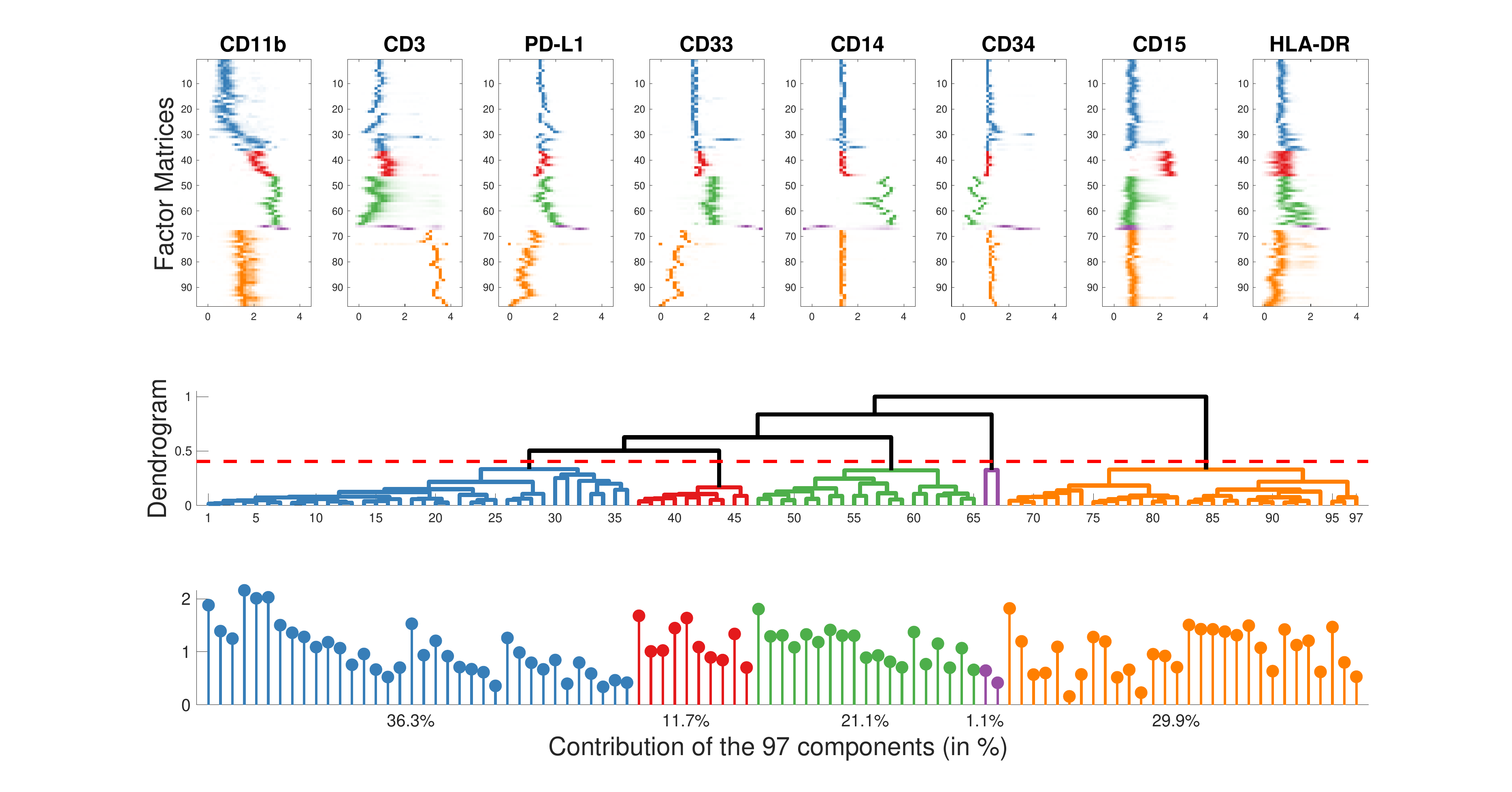}}
    \caption{PCTF3D applied to a graft sample dataset with 8 markers with a balanced coupling containing 14 triplets (one fourth of all possible triplets).}
    \label{fig:cytoAppli:g129hiercT4}
\end{figure}
\Cref{tab:cytoAppli:g129hierarc} shows that the PCTF3D and manual gating yield similar results.

\begin{table} 
    \centering
    \begin{tabular}{cccccc}
        \toprule
        \multicolumn{3}{c}{PCTF3D} & \multicolumn{3}{c}{Manual Gating} \\ \cmidrule(l{3pt}r{3pt}){1-3} \cmidrule(l{3pt}r{3pt}){4-6}
        Cluster & Fully & Balanced & Cell population & Marker expression & Proportion \\ \cmidrule(l{4pt}r{4pt}){1-3} \cmidrule(l{4pt}r{4pt}){4-6}
        Red & $12.4\%$ & $11.7\%$ & Granulocytes & CD15\textbf{+}& $10.7\%$ \\ 
        Green & $20.8\%$ & $21.1\%$ & MDSC & CD14\textbf{+}/HLA-DR\textbf{-} & $16\%$ \\ 
        Purple & $1.0\%$ & $1.1\%$ & Stem cells & CD34\textbf{+} & $0.61\%$ \\
        Orange & $30.5\%$ & $29.9\%$ & T cells & CD3\textbf{+} & $30.1\%$ \\
        Blue & $35.2\%$ & $36.3\%$ & Negative cells & - & $28.3\%$ 
        \\ \bottomrule
    \end{tabular}
    \caption{Identification of cell populations obtained with PCTF3D and manual gating (see Figures \ref{fig:cytoAppli:g129hierarc} and \ref{fig:cytoAppli:g129hiercT4}).
    MDSC stands for Myeloid-derived suppressor cells.}
    \label{tab:cytoAppli:g129hierarc}
\end{table}

\section{Conclusion}

In this work, a method for the fundamental statistical learning problem of joint PMF estimation is proposed.
This problem is compounded by the curse of dimensionality.
We here assume that the PMF follows a Naive Bayes Model: a graphical model, with both expressions in the probabilistic framework and in the context of tensor decompositions.
Benefiting the mild identifiability uniqueness of the CP decomposition, 
\cite{n_kargas_tensors_2018} proposed to perform high-dimensional joint PMF by using only 3-dimensional or 4-dimensional marginals, estimated accurately with fewer samples.
The novelty of our approach is to consider a subset of all possible 3D marginals (which can be extended to 4D marginals).
This approach was studied through the introduction of hypergraphs in the model to represent the coupling between marginals.
Choosing the said-coupling is tackled in this work, as two main strategies are proposed: random triplets and balanced couplings for which variables are evenly represented in the coupling.
On the practical side, our approach is suited for applications such as flow cytometry, as it gives interpretable results for biology experts with a non-destructive representation of a dataset.

This paper is the first part of a two-parts article.
After introducing a new coupled model in this Part I article, Part II \cite{pctf3d_part2} aims at studying uniqueness properties of PCTF3D's new model.
This problem is very important for both theoretical concerns but also for applications, as users want to ensure that results are reproducible among trials (which is surely not true for non-unique cases).
We will see that the model uniqueness highly depends on the coupling and more particularly on the structure of its hypergraph.

\section*{Acknowledgements}

The authors are very indebted to Dr. Anne-Béatrice Notarantonio for analyzing the flow cytometry datasets along with insightful expertise on the domain.

\appendix
\section{Algorithm updates} \label{app:updates}
\subsection{Update of factor matrices} 

The update of the $m$-th factor matrix $\Am{m}$ is provided thanks to \Cref{alg:upA}.
In this algorithm, the convergence criterion is the same as proposed in \cite{k_huang_flexible_2016}:
\begin{equation}
    r<\varepsilon \quad \text{and} \quad s<\varepsilon,
    \label{eq:convCrit}
\end{equation}
where
\begin{equation*}
    r = \frac{ \left\| \bfA_{t_2}-\bfB^\T_{t_2} \right\|^2_2}{\left\| \bfA_{t_2} \right\|^2_2}, 
\end{equation*}
and
\begin{equation*}
    s = \frac{ \left\| \bfA_{t_2}-\bfA_{t} \right\|^2_2}{\left\| \bfU_{t_2} \right\|^2_2}.
\end{equation*}

\begin{algorithm2e}
    \SetAlgoLined
    \KwIn{$\{\Am{m} \;|\; m\in\cpdsetp{1,M}\}$, $\;\lbd$, $\;\{\widetilde{\bcalH}^{(mk\ell)} \;|\; \triples{m}k\ell\in\calT\}$.}
    \KwSty{Initialization:}{ $\bfA=\Am{m}$, $\;\bfU\in\dsR^{I\times R}$ and $\widetilde{\bfA}\in\dsR^{R\times I}$ with zeros.}
    
      \For{ $\triples{m}k\ell\in\calT$}{
        $\bfQ^{(\ell k)} = \Am{\ell}\khatri\Am{k}$.
      }

    $\bfG = (\lbd\lbd^\T) \hadam \sum\limits_{\substack{k,\ell\\\triples{m}k\ell\in\calT}} \bfQ^{(\ell k)\T}\bfQ^{(\ell k)}$,

    $\bfW = \Diag(\lbd) \sum\limits_{\substack{k,\ell\\\triples{m}k\ell\in\calT}} \bfQ^{(\ell k)\T} \widetilde{\bcalH}^{(mk\ell)\T}_{(1)}$,

    $\rho = \frac{1}{R} \trace(\bfG)$.

    \While{convergence is not achieved}{
        $\widetilde{\bfA} \leftarrow \left(\bfG+\rho\bfI_R\right)^{-1}\left(\bfW+\rho\left(\bfA+\bfU\right)^\T\right)$,

        $\bfA \leftarrow \bfA-\widetilde{\bfA}^\T+\bfU$,

        Projection of $\bfA$ onto the simplex constraints \cite{duchi_efficient_2008},

        $\bfU \leftarrow \bfU+\bfA-\widetilde{\bfA}^{\T}$.
    }
    \KwOut{$\Am{m} = \bfA$.}
    \caption{Update of $\Am{m}$ (solving \eqref{eq:subOptimA}).}
    \label{alg:upA}
\end{algorithm2e}

\subsection{Update of the loading vector}
Similarly, the update of the loading vector $\lbd$ is given by the \Cref{alg:upL}.
The criterion is still \eqref{eq:convCrit} but the values of $r$ and $s$ changes:
\begin{equation*}
    r = \frac{ \left\| \bm\nu-\tilde{\bm\nu} \right\|^2_2}{\left\| \bm\nu \right\|^2_2} \quad \text{and} \quad
    s = \frac{ \left\| \bm\nu-\lbd \right\|^2_2}{\left\| \bfu \right\|^2_2}.
\end{equation*}
\begin{algorithm2e}
    \SetAlgoLined
    \KwIn{$\{\Am{m} \;|\; m\in\cpdsetp{1,M}\}$, $\;\lbd$, $\;\{\Hjklt \;|\; \jkl\in\calT \}$.}
    \KwSty{Initialization:}{ $\bm\nu=\lbd$, $\;\bfu\in\dsR^{R}$ and $\;\tilde{\bm\nu}\in\dsR^{R}$ with zeros.}
    
    \For{ $\jkl\in\calT$}{
        $\bfQ^{(\ell kj)} = \Am{\ell}\khatri\Am{k}\khatri \Am{j}$.
    }

    $\bfG = \sum\limits_{\jkl\in\calT} \bfQ^{(\ell kj)\T}\bfQ^{(\ell kj)}$,

    $\bfw = \sum\limits_{\jkl\in\calT} \bfQ^{(\ell kj)\T} \vectorize(\Hjklt)$, 
 
    $\rho = \frac{1}{R} \trace(\bfG)$,

    \While{convergence is not achieved}{
        $\tilde{\bm\nu} \leftarrow \left(\bfG+\rho\bfI_R\right)^{-1}\left(\bfw+\rho(\bm\nu+\bfu)\right)$,

        $\bm\nu \leftarrow \bm\nu-\tilde{\bm\nu}+\bfu$,

        Projection of $\bm\nu$ onto the simplex constraints \cite{duchi_efficient_2008},

        $\bfu \leftarrow \bfu+\bm\nu-\tilde{\bm\nu}$.
    }
    \KwOut{$\lbd = \bm\nu$.}
    \caption{$\lbd$ update (solving \eqref{eq:subOptimL}).}
    \label{alg:upL}
\end{algorithm2e}

\bibliographystyle{elsarticle-num}

\begin{thebibliography}{}
\expandafter\ifx\csname url\endcsname\relax
  \def\url#1{\texttt{#1}}\fi
\expandafter\ifx\csname urlprefix\endcsname\relax\def\urlprefix{URL }\fi
\expandafter\ifx\csname href\endcsname\relax
  \def\href#1#2{#2} \def\path#1{#1}\fi

\end{thebibliography}


\begin{thebibliography}{10}
\expandafter\ifx\csname url\endcsname\relax
  \def\url#1{\texttt{#1}}\fi
\expandafter\ifx\csname urlprefix\endcsname\relax\def\urlprefix{URL }\fi
\expandafter\ifx\csname href\endcsname\relax
  \def\href#1#2{#2} \def\path#1{#1}\fi

\bibitem{duda_pattern_1973}
R.~O. Duda, P.~E. Hart, Pattern classification and scene analysis, Vol.~3, Wiley New York, 1973.

\bibitem{van_trees_detection_2004}
H.~L. Van~Trees, Detection, estimation, and modulation theory, part {I}: detection, estimation, and linear modulation theory, John Wiley \& Sons, 2004.

\bibitem{n_kargas_tensors_2018}
{N. Kargas}, {N. D. Sidiropoulos}, {X. Fu}, Tensors, {Learning}, and “{Kolmogorov} {Extension}” for {Finite}-{Alphabet} {Random} {Vectors}, IEEE Transactions on Signal Processing 66~(18) (2018) 4854--4868, number: 18.
\newblock \href {https://doi.org/10.1109/TSP.2018.2862383} {\path{doi:10.1109/TSP.2018.2862383}}.

\bibitem{kanatsoulis_hyperspectral_2018}
C.~I. Kanatsoulis, X.~Fu, N.~D. Sidiropoulos, W.-K. Ma, Hyperspectral super-resolution: A coupled tensor factorization approach, IEEE Transactions on Signal Processing 66~(24) (2018) 6503--6517.

\bibitem{prevost_hyperspectral_2022}
C.~Pr{\'e}vost, R.~A. Borsoi, K.~Usevich, D.~Brie, J.~C. Bermudez, C.~Richard, Hyperspectral super-resolution accounting for spectral variability: Coupled tensor ll1-based recovery and blind unmixing of the unknown super-resolution image, SIAM Journal on Imaging Sciences 15~(1) (2022) 110--138.

\bibitem{borsoi:hal-04135339}
R.~A. Borsoi, I.~Lehmann, M.~A. B.~S. Akhonda, V.~Calhoun, K.~Usevich, D.~Brie, T.~Adali, \href{https://hal.science/hal-04135339}{{Coupled CP tensor decomposition with shared and distinct components for multi-task Fmri data fusion}}, in: {IEEE International Conference on Acoustics, Speech and Signal Processing (ICASSP)}, Rhodes Island, Greece, 2023, pp. 1--5.
\newline\urlprefix\url{https://hal.science/hal-04135339}

\bibitem{ermics2015link}
B.~Ermi{\c{s}}, E.~Acar, A.~T. Cemgil, Link prediction in heterogeneous data via generalized coupled tensor factorization, Data Mining and Knowledge Discovery 29 (2015).

\bibitem{papalexakis2016tensors}
E.~E. Papalexakis, C.~Faloutsos, N.~D. Sidiropoulos, Tensors for data mining and data fusion: Models, applications, and scalable algorithms, ACM Transactions on Intelligent Systems and Technology (TIST) 8~(2) (2016) 1--44.

\bibitem{ibrahim_recovering_2021}
S.~Ibrahim, X.~Fu, Recovering joint probability of discrete random variables from pairwise marginals, IEEE Transactions on Signal Processing 69 (2021) 4116--4131, iEEE.

\bibitem{a_yeredor_maximum_2019}
{A. Yeredor}, {M. Haardt}, Maximum {Likelihood} {Estimation} of a {Low}-{Rank} {Probability} {Mass} {Tensor} {From} {Partial} {Observations}, IEEE Signal Processing Letters 26~(10) (2019) 1551--1555, number: 10.
\newblock \href {https://doi.org/10.1109/LSP.2019.2938663} {\path{doi:10.1109/LSP.2019.2938663}}.

\bibitem{kargas2019learning}
N.~Kargas, N.~D. Sidiropoulos, Learning mixtures of smooth product distributions: Identifiability and algorithm, in: The 22nd International Conference on Artificial Intelligence and Statistics, PMLR, 2019, pp. 388--396.

\bibitem{reynolds_gaussian_2009}
D.~A. Reynolds, Gaussian mixture models., Encyclopedia of biometrics 741~(659-663), berlin, Springer (2009).

\bibitem{hsu_learning_2013}
D.~Hsu, S.~M. Kakade, Learning mixtures of spherical gaussians: moment methods and spectral decompositions, in: Proceedings of the 4th conference on {Innovations} in {Theoretical} {Computer} {Science}, 2013, pp. 11--20.

\bibitem{nowak_distributed_2003}
R.~D. Nowak, Distributed {EM} algorithms for density estimation and clustering in sensor networks, IEEE transactions on signal processing 51~(8) (2003) 2245--2253, iEEE.

\bibitem{gyllenberg_non-uniqueness_1994}
M.~Gyllenberg, T.~Koski, E.~Reilink, M.~Verlaan, Non-uniqueness in probabilistic numerical identification of bacteria, Journal of Applied Probability 31~(2) (1994) 542--548, cambridge University Press.

\bibitem{nowicki_estimation_2001}
K.~Nowicki, T.~A.~B. Snijders, Estimation and prediction for stochastic blockstructures, Journal of the American statistical association 96~(455) (2001) 1077--1087, taylor \& Francis.

\bibitem{pctf3d_part2}
P.~Flores, K.~Usevich, D.~Brie, Coupled tensor models for probability mass function estimation: {Part} {II}, {Uniqueness} of the model., Submitted to Signal processing (jul 2025).

\bibitem{kolda_tensor_2009}
T.~Kolda, B.~Bader, Tensor {Decompositions} and {Applications}, SIAM Review 51 (2009) 455--500.
\newblock \href {https://doi.org/10.1137/07070111X} {\path{doi:10.1137/07070111X}}.

\bibitem{comon_tensors_2014}
P.~Comon, Tensors: a {Brief} {Introduction}, Signal Processing Magazine, IEEE 31 (2014) 44--53.
\newblock \href {https://doi.org/10.1109/MSP.2014.2298533} {\path{doi:10.1109/MSP.2014.2298533}}.

\bibitem{hitchcock_expression_1927}
F.~L. Hitchcock, The expression of a tensor or a polyadic as a sum of products, Journal of Mathematics and Physics 6~(1-4) (1927) 164--189, wiley Online Library.

\bibitem{bellman_adaptive_1959}
R.~Bellman, R.~Kalaba, On adaptive control processes, IRE Transactions on Automatic Control 4~(2) (1959) 1--9, iEEE.

\bibitem{jordan_graphical_2004}
M.~I. Jordan, {Graphical Models}, Statistical Science 19~(1) (2004) 140 -- 155.

\bibitem{ishteva_tensors_2015}
M.~Ishteva, Tensors and latent variable models, in: Latent Variable Analysis and Signal Separation: 12th International Conference, LVA/ICA 2015, Liberec, Czech Republic, August 25-28, 2015, Proceedings 12, Springer, 2015, pp. 49--55.

\bibitem{robeva_duality_2019}
E.~Robeva, A.~Seigal, Duality of graphical models and tensor networks, Information and Inference: A Journal of the IMA 8~(2) (2019) 273--288, oxford University Press.

\bibitem{k_huang_flexible_2016}
{K. Huang}, {N. D. Sidiropoulos}, {A. P. Liavas}, A {Flexible} and {Efficient} {Algorithmic} {Framework} for {Constrained} {Matrix} and {Tensor} {Factorization}, IEEE Transactions on Signal Processing 64~(19) (2016) 5052--5065, number: 19.
\newblock \href {https://doi.org/10.1109/TSP.2016.2576427} {\path{doi:10.1109/TSP.2016.2576427}}.

\bibitem{boyd_distributed_2011}
S.~Boyd, N.~Parikh, E.~Chu, B.~Peleato, J.~Eckstein, Distributed optimization and statistical learning via the alternating direction method of multipliers, Foundations and Trends® in Machine learning 3~(1) (2011) 1--122, now Publishers, Inc.

\bibitem{frosini_degree_2013}
A.~Frosini, C.~Picouleau, S.~Rinaldi, On the degree sequences of uniform hypergraphs, in: Discrete {Geometry} for {Computer} {Imagery}: 17th {IAPR} {International} {Conference}, {DGCI} 2013, {Seville}, {Spain}, {March} 20-22, 2013. {Proceedings} 17, Springer, 2013, pp. 300--310.

\bibitem{boonyasombat_degree_1984}
V.~Boonyasombat, Degree sequences of connected hypergraphs and hypertrees, in: Graph {Theory} {Singapore} 1983, Vol. 1073, Springer Berlin Heidelberg, Berlin, Heidelberg, 1984, pp. 236--247, series Title: Lecture Notes in Mathematics.

\bibitem{harzheim_ordered_2005}
E.~Harzheim, Ordered sets, Vol.~7, Springer Science \& Business Media, 2005.

\bibitem{gilbert_symmetry_1961}
E.~N. Gilbert, J.~Riordan, Symmetry types of periodic sequences, Illinois Journal of Mathematics 5~(4) (1961) 657--665, duke University Press.

\bibitem{acar_scalable_2011}
E.~Acar, D.~M. Dunlavy, T.~G. Kolda, M.~Mørup, Scalable tensor factorizations for incomplete data, Chemometrics and Intelligent Laboratory Systems 106~(1) (2011) 41--56, elsevier.

\bibitem{flores_probability_2023}
P.~Flores, J.~K. Chege, K.~Usevich, M.~Haardt, A.~Yeredor, D.~Brie, Probability mass function estimation approaches with application to flow cytometry data analysis, in: 2023 IEEE 9th International Workshop on Computational Advances in Multi-Sensor Adaptive Processing (CAMSAP), IEEE, 2023, pp. 451--455.

\bibitem{vesey_detection_1994}
G.~Vesey, J.~Narai, N.~Ashbolt, K.~Williams, D.~Veal, Detection of specific microorganisms in environmental samples using flow cytometry, in: Methods in cell biology, Vol.~42, Elsevier, 1994, pp. 489--522.

\bibitem{seidel_jr_sexing_2003}
G.~E. Seidel~Jr, Sexing mammalian sperm—intertwining of commerce, technology, and biology, Animal Reproduction Science 79~(3-4) (2003) 145--156, elsevier.

\bibitem{greve_flow_2012}
B.~Greve, R.~Kelsch, K.~Spaniol, H.~T. Eich, M.~Götte, Flow cytometry in cancer stem cell analysis and separation, Cytometry Part A 81~(4) (2012) 284--293, wiley Online Library.

\bibitem{daveni_myeloid-derived_2020}
M.~d'Aveni, A.~B. Notarantonio, A.~Bertrand, L.~Boulangé, C.~Pochon, M.~T. Rubio, Myeloid-derived suppressor cells in the context of allogeneic hematopoietic stem cell transplantation, Frontiers in immunology 11 (2020) 989, frontiers Media SA.

\bibitem{chattopadhyay_good_2010}
P.~K. Chattopadhyay, M.~Roederer, Good cell, bad cell: {Flow} cytometry reveals {T}-cell subsets important in {HIV} disease, Cytometry part A 77~(7) (2010) 614--622, wiley Online Library.

\bibitem{perfetto_seventeen-colour_2004}
S.~P. Perfetto, P.~K. Chattopadhyay, M.~Roederer, Seventeen-colour flow cytometry: unravelling the immune system, Nature Reviews Immunology 4~(8) (2004) 648--655, nature Publishing Group UK London.

\bibitem{park2020omip}
L.~M. Park, J.~Lannigan, M.~C. Jaimes, Omip-069: forty-color full spectrum flow cytometry panel for deep immunophenotyping of major cell subsets in human peripheral blood, Cytometry Part A 97~(10) (2020) 1044--1051.

\bibitem{lahnemann_eleven_2020}
D.~Lähnemann~et al., Eleven grand challenges in single-cell data science, Genome Biology 21~(1) (2020) 31.
\newblock \href {https://doi.org/10.1186/s13059-020-1926-6} {\path{doi:10.1186/s13059-020-1926-6}}.

\bibitem{milligan_ultrametric_1979}
G.~W. Milligan, Ultrametric hierarchical clustering algorithms, Psychometrika 44~(3) (1979) 343--346, springer.

\bibitem{duchi_efficient_2008}
J.~Duchi, S.~Shalev-Shwartz, Y.~Singer, T.~Chandra, Efficient projections onto the $\ell$1-ball for learning in high dimensions, Proceedings of the 25th International Conference on Machine Learning (2008) 272--279\href {https://doi.org/10.1145/1390156.1390191} {\path{doi:10.1145/1390156.1390191}}.

\end{thebibliography}

\end{document}